%% file: main.tex
\documentclass[a4paper,UKenglish,cleveref, autoref, thm-restate, svgnames]{lipics-v2021}

\hideLIPIcs  

\usepackage{amssymb}
\input{preamble}

\title{Romeo and Juliet Meeting in Forest Like Regions} 

\titlerunning{Romeo and Juliet Meeting in Forest Like Regions} 

\author{Neeldhara Misra}{Indian Institute of Technology, Gandhinagar, India \and \url{https://www.neeldhara.com} }{neeldhara.m@iitgn.ac.in}{https://orcid.org/0000-0003-1727-5388}{The author is grateful for support from DST-SERB and IIT Gandhinagar. This work was partially supported by the ECR grant ECR/2018/002967.}

\author{Manas Mulpuri}{Indian Institute of Technology, Gandhinagar, India}{mulpuri.m@iitgn.ac.in}{}{}

\author{Prafullkumar Tale}{Indian Institute of Science Education and Research, Pune, India\and \url{https://pptale.github.io/}}{prafullkumar@iiserpune.ac.in}{}{Part of the work was carried out when the author was a Post-Doctoral Researcher at CISPA Helmholtz Center for Information Security, Germany,  supported by the European Research Council (ERC) consolidator grant No.~725978 SYSTEMATICGRAPH.}

\author{Gaurav Viramgami}{Indian Institute of Technology, Gandhinagar, India}{viramgami.g@iitgn.ac.in}{}{}

\authorrunning{Misra, Mulpuri, Tale, and Viramgami} 

\Copyright{Neeldhara, Manas, Prafullkumar, and Gaurav} 



\category{} 

\relatedversion{A shorter version of this work has been accepted for presentation at the 42nd IARCS Annual Conference on Foundations of Software Technology and Theoretical Computer Science (FSTTCS), 2022.} 

\supplement{}

\acknowledgements{We are grateful for feedback from anonymous reviewers.}

\nolinenumbers 

\EventEditors{John Q. Open and Joan R. Access}
\EventNoEds{2}
\EventLongTitle{42nd Conference on Very Important Topics (CVIT 2016)}
\EventShortTitle{CVIT 2016}
\EventAcronym{CVIT}
\EventYear{2016}
\EventDate{December 24--27, 2016}
\EventLocation{Little Whinging, United Kingdom}
\EventLogo{}
\SeriesVolume{42}
\ArticleNo{23}

\begin{document}

\maketitle

\begin{abstract}
The game of rendezvous with adversaries is a game on a graph played by two players: \emph{Facilitator} and \emph{Divider}.
Facilitator has two agents and Divider has a team of $k \ge 1$ agents.
{While the initial positions of Facilitator's agents are fixed, Divider gets to select the initial positions of his agents.}
{Then, they take turns to move their agents to adjacent vertices (or stay put) with Facilitator's goal to bring both her agents at same vertex and Divider's goal to prevent it.}
The computational question of interest is to determine if Facilitator has a winning strategy against Divider with $k$ agents.
Fomin, Golovach, and Thilikos [WG, 2021] introduced this game and proved that it is \PSPACE-\hard\ and \co-\W[2]-\hard\ parameterized by the number of agents.

This hardness naturally motivates the structural parameterization of the problem.
The authors proved that it admits an \FPT\ algorithm when parameterized by the modular width and the number of allowed rounds.
However, they left open the complexity of the problem from the perspective of other structural parameters.
In particular, they explicitly asked whether the problem admits an \FPT\ or \XP-algorithm with respect to the treewidth of the input graph.
We answer this question in the negative and show that \textsc{Rendezvous} is \co-\NP-\hard\ even for graphs of constant treewidth.
Further, we show that the problem is \co-\W[1]-\hard\ when parameterized by the feedback vertex set number and the number of agents, and is unlikely to admit a polynomial kernel when parameterized by the vertex cover number and the number of agents.
Complementing these hardness results, we show that the \textsc{Rendezvous} is \FPT\ when parameterized by both the vertex cover number and the solution size.
Finally, for graphs of treewidth at most two {and girds}, we show that the problem can be solved in polynomial time.
\end{abstract}

\newpage

\input{introduction}

\input{preliminaries}
\input{tree-width}
\input{appendix}
\input{feedback-vertex-set}
\input{vertex-cover}
\input{vc-npk}
\input{polynomial-cases}
\input{conclusion}

\bibliography{references}



\end{document}

%% file: preamble.tex
\usepackage{amsmath}
\usepackage{amsthm}
\usepackage{graphicx}
\usepackage[linesnumbered,algoruled,boxed,lined]{algorithm2e}

\usepackage{complexity} 

\usepackage{mathtools}
\usepackage{thmtools}
\usepackage[shortlabels]{enumitem}
\setlist{nosep}

\usepackage{tikz}
\usetikzlibrary{positioning,backgrounds,patterns,calc}

\tikzset{
  circ/.style = {circle,draw,fill,inner sep=1.3pt}
  }  

\usepackage{comment}

\includecomment{fullversion}
\excludecomment{shortversion}

\excludecomment{ournotes}

\includecomment{cond}
\newcommand{\maybe}[1]{}
\begin{cond}
\renewcommand{\maybe}[1]{#1}
\end{cond}

\usepackage{xcolor}

\usepackage{amssymb}

\usepackage[most]{tcolorbox}

\newcommand{\defproblemboxed}[3]{
\vspace{1mm}
\begin{center}
\begin{tcolorbox}[colback=WhiteSmoke,colframe=DarkSlateBlue,title=\textsc{#1}]
  {\bf{Input:}} #2  \\
  {\bf{Question:}} #3
\end{tcolorbox}
\end{center}
\vspace{1mm}
}

\usepackage{color}

\newcommand{\tw}{\mathtt{tw}}
\newcommand{\pw}{\mathtt{pw}}
\newcommand{\vc}{\mathtt{vc}}
\newcommand{\fvs}{\mathtt{fvs}}

\newcommand{\degree}{\mathtt{deg}}

\newcommand{\para}{\textsf{para}}

\newcommand{\calC}{\mathcal{C}}
\newcommand{\calD}{\mathcal{D}}
\newcommand{\calF}{\mathcal{F}}

\newcommand{\calO}{\ensuremath{{\mathcal O}}}
\newcommand{\calP}{\mathcal{P}}

\newcommand{\calT}{\mathcal{T}}
\newcommand{\calU}{\mathcal{U}}

\newcommand{\calX}{\mathcal{X}}

\newcommand{\true}{\texttt{True}}
\newcommand{\false}{\texttt{False}}

\newcommand{\hard}{\textsf{hard}}
\newcommand{\paraNP}{\textsf{paraNP}}

\newcommand{\yes}{\textsc{Yes}}
\newcommand{\no}{\textsc{No}}

\newtheorem{reduction rule}{Reduction Rule}[section]
\newtheorem{marking-scheme}{Marking Scheme}[section]

\makeatletter
\patchcmd{\BR@backref}{\newblock}{\newblock($\uparrow$~}{}{}
\patchcmd{\BR@backref}{\par}{)\par}{}{}
\makeatother

%% file: introduction.tex
\section{Introduction}

The game of rendezvous with adversaries on a graph --- \textsc{Rendezvous} --- is a natural dynamic version of the problem of finding a vertex cut between two vertices $s$ and $t$ introduced by Fomin, Golovach, and Thilikos~\cite{DBLP:conf/wg/FominGT21}.
The game is played on a finite undirected connected graph $G$ by two players: \emph{Facilitator} and \emph{Divider}.
Facilitator has two agents Romeo and Juliet that are initially placed in designated vertices $s$ and $t$ of $G$.
Divider, on the other hand, has a team of $k \ge 1$ agents $D_1, \ldots, D_k$ that are initially placed in some vertices of $V(G)\backslash\{s, t\}$ chosen by him.
We note that a single vertex can accommodate multiple agents of Divider.

Then the players make their moves by turn, starting with Facilitator.
At every move, each player moves some of his/her agents to adjacent vertices or keeps them in their old positions.
No agent can be moved to a vertex that is currently occupied by adversary's agents.
Both players have complete information about $G$ and the positions of all the agents.
Facilitator aims to ensure that Romeo and Juliet meet; that is, they are in the same vertex.
The task of Divider is to prevent the rendezvous of Romeo and Juliet by maintaining $D_1, \ldots, D_k$ in positions that block the possibility to meet.
Facilitator wins if Romeo and Juliet meet, and Divider wins if they succeed in preventing the meeting of Romeo and Juliet forever. This setup naturally leads to the following computational question.

\defproblemboxed{Rendezvous}{A graph $G$ with two given vertices $s$ and $t$, and a positive integer $k$.}{Can Facilitator win on $G$ starting from $s$ and $t$ against Divider with $k$ agents?}

We will often refer to $k$, the number of agents employed by Divider to keep Romeo and Juliet separated, as the ``solution size'' for this problem.

\paragraph*{Known Results}

Fomin, Golovach, and Thilikos~\cite{DBLP:conf/wg/FominGT21} initiated an extensive study of the computational complexity of \textsc{Rendezvous}.
They concluded that the problem is \PSPACE-\hard\ and \co-\W[2]-\hard\footnote{We refer the reader to~\Cref{sec:prelim} for the definitions of these complexity classes.} when parameterized by the number of Divider's agents, while also demonstrating an $|V(G)|^{\mathcal{O}(k)}$ algorithm based on backtracking stages over the game arena.
They also show that the problem admits polynomial time algorithms on chordal graphs and $P_5$-free graphs.
A related problem considered is~\textsc{Rendezvous in Time}, which asks if Facilitator can force a win in at most $\tau$ steps.
It turns out that~\textsc{Rendezvous in Time} is co-NP-complete even for $\tau = 2$ and is \FPT\ when parameterized by $\tau$ and the neighborhood diversity of the graph.
The latter is an ILP-based approach and uses the fact that \textsc{Integer Linear Programming Feasibility} is \FPT\ in the number of variables.
{We refer readers to \cite{DBLP:conf/wg/FominGT21}, and references within, for more related problems.}

The smallest number of agents that Divider needs to use to win on a graph $G$ is called the ``dynamic'' separation number of $G$.
We denote this by $d_G(s,t)$.
Note that if $s$ and $t$ are adjacent or $s = t$, then $d_G(s,t) := + \infty$.
The ``static'' separation number between $s$ and $t$, the original positions of Facilitator's agents, is simply the smallest size of a $(s,t)$-vertex cut, i.e, a subset of vertices whose removal disconnects $s$ and $t$.
We use $\lambda_G(s,t)$ to denote the minimum size of a vertex $(s,t)$-separator in $G$.
It is clear that $d_G(s,t) \leq \lambda_G(s,t)$, since positioning $\lambda_G(s,t)$ many guards on the vertices of a $(s,t)$-vertex allows Divider to win the game right away.
It turns out that $d_G(s,t) = 1$ if and only if $\lambda_G(s,t) = 1$.
However, there are examples of graphs where $d_G(s,t)$ is arbitrarily smaller than $\lambda_G(s,t)$~\cite{DBLP:conf/wg/FominGT21}.
The results in~\cite{DBLP:conf/wg/FominGT21} for chordal graphs and $P_5$-free graphs are based on the fact that in these graphs, it turns out that $d_G(s,t) = \lambda_G(s,t)$.

\paragraph*{Our Contributions}

Given that the problem is hard in the solution size, often regarded the ``standard'' parameter, a natural approach is to turn to structural parameters of the input graph.
One of the most popular structural parameters in the context of graphs is treewidth, which is a measure of how ``tree-like'' a graph is.
XP and FPT algorithms parameterized by treewidth are natural generalizations of tractability on trees.
Indeed, \textsc{Rendezvous} is easy to solve on trees because $\lambda_G(s,t) = 1$ for any distinct $s$ and $t$ when $G$ is a tree and $st \notin E(G)$.
The complexity of \textsc{Rendezvous} parameterized by treewidth, however, is wide open --- in particular it is not even known if the problem is in XP parameterized by treewidth.

Interestingly, it was pointed out in~\cite{DBLP:conf/wg/FominGT21} that if the initial positions $s$ and $t$ are not is the same bag of a tree decomposition of width $w$, then the upper bound for the dynamic separation number by $\lambda_G(s, t)$ together with the XP algorithm for the standard parameter can be employed to solve the problem in $n^{O(w)}$ time.
Thus, the question that was left open by Fomin, Golovach, and Thilikos was if the problem can be solved in the same time if $s$ and $t$ are in the same bag.
Our first contribution is to answer this question in the negative by showing that \textsc{Rendezvous} is in fact \co-\NP-\hard\ even for graphs of \emph{constant} treewidth. In fact, we show more:

\begin{restatable}{theorem}{fvsconphard}
    \label{thm:fvs-conp-hard}
    {\sc Rendezvous} is {\em \co-\NP-\hard} even when restricted to:
    \begin{itemize}
        \item graphs whose feedback vertex set number is at most 14, or
        \item graphs whose pathwidth is at most 16.
    \end{itemize}
    In particular, {\sc Rendezvous} is {\em \co-\para-\NP-\hard} parameterized by treewidth.
\end{restatable}

We obtain this hardness by a non-trivial reduction from the \textsc{3-Dimensional Matching} problem.
In the backdrop of this somewhat surprising result, we are motivated to pursue the question of the complexity of \textsc{Rendezvous} for larger parameters.
It turns out that even augmenting the feedback vertex set {number or the pathwidth} with the solution size is not enough.
Specifically, we show that the problem is {unlikely to admit an \FPT-algorithm even} when parameterized by {these combined parameters}.

\begin{restatable}{theorem}{fvswhard}
    \label{thm:fvs-w-hard}
    {\sc Rendezvous} is {\em \co-\W[1]-\hard} when parameterized by:
    \begin{itemize}
        \item the feedback vertex set number and the solution size, or
        \item the pathwidth and the solution size.
    \end{itemize}
\end{restatable}

This result is shown by a parameter preserving reduction from the \textsc{(Monotone) NAE-Integer-3-Sat} problem, which was shown to be {\W[1]-}\hard\ when parameterized by the number of variables by Bringmann et al.~\cite{DBLP:journals/jcss/BringmannHML16}.
Note that with this, we have a reasonably complete understanding of {\sc Rendezvous} in the combined parameter.
Indeed, recall that the problem is \co-\W[1]-\hard\ and \XP\ parameterized by the solution size alone, and \co-\para-\NP-\hard{} parameterized by the feedback vertex set number alone as shown above.

Given the above hardness, we consider \textsc{Rendevous} parameterized by the vertex cover number, a larger parameter compared to both the feedback vertex set number and pathwidth.
The status of~\textsc{Rendevous} with respect to the vertex cover parameterization was also left open in~\cite{DBLP:conf/wg/FominGT21}.
We see that the problem admits a natural exponential kernel in this parameter when combined with the solution size, and is hence FPT in the combined parameter; however this kernel cannot be improved to a polynomial kernel under standard complexity-theoretic assumptions.
\begin{restatable}{theorem}{vcfptnpk}
    \label{thm:vc-fpt-no-poly}
    {\sc Rendezvous} is \FPT\ when parameterized by the vertex cover number of the input graph and the solution size.
    Moreover, the problem does not admit a polynomial kernel when parameterized by the vertex cover number and the solution size unless \NP $\subseteq$ \co-\NP/poly.
\end{restatable}
We briefly describe the intuition for the exponential kernel with respect to the vertex cover number.
Suppose the graph $G$ has a vertex cover $X$, where $|X| \leq \ell$, and, one may assume, without loss of generality, that $s,t \in X$.
Further, for any $Y \subseteq X$, let $I_Y$ denote the set of all vertices in $G \setminus X$ whose neighborhood in $X$ is exactly $Y$.
It is not hard to see that if $|I_Y| > k$, then one might as well ``curtail'' the set to $k+1$ vertices without changing the instance. This leads to an exponential kernel in the combined parameter.
It is also true that $k$ is bounded, without loss of generality, by $\ell$ and the size of the common neighbhorhood of $s$ and $t$ to begin with; however, it is unclear if $k$ can always be bounded by some function of the vertex cover alone.
The kernelization lower bound follows from observing the structure of the reduced instance in {the} reduction {used} in~\cite{DBLP:conf/wg/FominGT21} {to prove that problem is \co-\W[2]-\hard\ when parameterized by the solution size}.

Finally, we present polynomial time algorithms on two restricted cases.
\begin{restatable}{theorem}{polytimecases}
    \label{thm:poly-time-cases}
    {\sc Rendezvous} can be solved in polynomial time on the classes of treewidth at most two graphs and grids.
\end{restatable}
Recall that the polynomial time algorithm on the classes of trees, chordal graphs, and $P_5$-free graphs is obtained by proving  that the size of dynamic separator is same as that of separator.
In case of grids, we present a winning strategy for Divider for any non-trivial instances.
This makes grids unique graph class in which the problem admits polynomial time algorithm even when dynamic separator can be smaller than separator.

\textbf{Organization of the paper.}
After presenting technical preliminaries in~\Cref{sec:prelim}, we first describe the proof of~\Cref{thm:fvs-conp-hard} in~\Cref{sec:tree-width}, along with a separate discussion focused on the intuition for the proof.
We present \Cref{thm:fvs-w-hard} in \Cref{sec:feedback-vertex-set}.
The proof of~\Cref{thm:vc-fpt-no-poly} can be found in~\Cref{sec:vertex-cover}.
The polynomial time results are presented in~\Cref{sec:polycases}.




%% file: preliminaries.tex
\section{Preliminaries}
\label{sec:prelim}

For a positive integer $q$, we denote {the} set $\{1, 2, \dots, q\}$ by $[q]$.
We use $\mathbb{N}$ to denote the collection of all non-negative integers.

\paragraph*{Graph theory}
We use standard graph-theoretic notation, and we refer the reader to~\cite{DBLP:books/daglib/Diestel12} for any undefined notation.
For an undirected graph $G$, sets $V(G)$ and $E(G)$ denote its set of vertices and edges, respectively.
We denote an edge with two endpoints $u, v$ as $uv$. Unless otherwise specified, we use $n$ to denote the number of vertices in the input graph $G$ of the problem under consideration.
Two vertices $u, v$ in $V(G)$ are \emph{adjacent} if there is an edge $uv$ {in $G$}.
The \emph{open neighborhood} of a vertex $v$, denoted by $N_G(v)$, is the set of vertices adjacent to $v$.
The \emph{closed neighborhood} of a vertex $v$, denoted by $N_G[v]$, is the set $N_G(v) \cup \{v\}$.
We say that a vertex $u$ is a \emph{pendant vertex} if $|N_G(v)| = 1$.
The \emph{degree} of a vertex $v$, denoted by $\degree_{G}(v)$, is equal to the number of vertices in the open neighbourhood of $v$, i.e., $\degree_{G}(v) = |N_{G}(v)|$.
We omit the subscript in the notation for neighborhood if the graph under consideration is clear.

For a subset $S$ of $V(G)$, we define $N[S] = \bigcup_{v \in S} N[v]$ and $N(S) = N[S] \setminus S$.
For a subset $F$ of edges, we denote by $V(F)$ the collection of endpoints of edges in $F$.
For a subset $S$ of $V(G)$ ({resp.} a subset $F$ of $E(G)$), we denote the graph obtained by deleting $S$ ({resp.} deleting $F$) from $G$ by $G - S$ ({resp.} by $G - F$).
We denote the subgraph of $G$ induced on the set $S$ by $G[S]$.

A graph is {\em connected} if there is a path between every pair of distinct vertices.
A subset $S \subseteq V(G)$ is said to be a \emph{connected set} if $G[S]$ is connected.

A \emph{simple path}, denoted by $P[u, v, d]$, is a non-empty graph $G$ of the form $V(G) = \{u, x_{1}, \ldots, x_{d}, v\}$, and $E(G) = \{ux_{1}, x_{1}x_{2}, \ldots, x_{d-1}x_{d}, x_{d}v\}$, where $u$, $v$, and all $x_{i}$'s are distinct.
The vertices $\{x_{1}, x_{2}, \ldots, x_{d}\}$ are the \emph{internal vertices} of $P[u, v, d]$, and
the vertices $\{x_{i} : \degree(x_{i}) > 2, i \in [d]\}$, i.e., internal vertices whose \emph{degree} is strictly greater than $2$ are the \emph{branching points} of $P[u, v, d]$.
We use $P[u, v, d_1] \circ P[v, w, d_2]$ to denote the unique simple path from $u$ to $w$ that contains $v$ and has $d_1 + d_2 + 1$ many internal vertices.

A set of vertices $Y$ is said to be {an} \emph{independent set} if no two vertices in $Y$ are adjacent.
For a graph $G$, a set $X \subseteq V(G)$ is said to be {a} \emph{vertex cover} if $V(G) \setminus X$ is an independent set.
A set of vertices $Y$ is said to be a \emph{clique} if any two vertices in $Y$ are adjacent.
A vertex cover $X$ is a \emph{minimum vertex cover} if for any other vertex cover $Y$ of $G$, we have $|X| \le |Y|$.
We denote by $\vc(G)$ the size of {a} minimum vertex cover of {a graph} $G$.
For a graph $G$, a set $X \subseteq V(G)$ is said to be {a} \emph{feedback vertex set} if $V(G) \setminus X$ is does not contain a cycle.
We denote by $\fvs(G)$ the size of {a} minimum feedback vertex set of {a graph} $G$.

A \emph{path decomposition} of a graph $G$ is a sequence $\calP = (X_{1}, X_{2}, \ldots, X_{r})$ of \emph{bags}, where $X_{i} \subseteq V(G)$ for each $i \in [r]$, such that the following conditions hold:
\begin{itemize}
    \item $\bigcup_{i = 1}^{r} X_{i} = V(G)$.
    In other words, every vertex of $G$ is in at least one bag.
    \item For every $uv \in E(G)$, there exists $\ell \in [r]$ such that the bag $X_{\ell}$ contains both $u$ and $v$.
    \item For every $u \in V(G)$, if $u \in X_{i} \cap X_{k}$ for some $i \le k$, then $u \in X_{j}$ also for each $j$ such that $i \le j \le k$.
    In other words, the indices of the bags containing $u$ form an interval in $[r]$.
\end{itemize}
The \emph{width} of a path decomposition $(X_{1}, X_{2}, \ldots, X_{r})$ is $max_{1 \le i \le r} |X_{i}| - 1$.
The \emph{pathwidth} of a graph $G$, denoted by $\pw(G)$, is the minimum possible width of a path decomposition of $G$.

A \emph{tree decomposition} of a graph $G$ is a pair $\calT = (T, \{X_{t}\}_{t \in V(T)})$, where $T$ is a tree whose every node $t$ is assigned a vertex subset $X_{t} \subseteq V(G)$, called a \emph{bag}, such that the following conditions hold:
\begin{itemize}
    \item $\bigcup_{t \in V(T)} X_{t} = V(G)$.
    In other words, every vertex of $G$ is in at least one bag.
    \item For every $uv \in E(G)$, there exists a node $t$ of $T$ such that bag $X_{t}$ contains both $u$ and $v$.
    \item For every $u \in V(G)$, the set $T_{u} = \{t \in V(T) : u \in X_{t}\}$, i.e., the set of nodes whose corresponding bags contains $u$, induces a connected subtree of $T$.
\end{itemize}
The \emph{width} of a tree decomposition $\calT = (T, \{X_{t}\}_{t \in V(T)})$ is $max_{t \in V(T)} |X_{t}| - 1$.
The \emph{treewidth} of a graph $G$, denoted by $\tw(G)$, is the minimum possible width of a tree decomposition of $G$.

A $M \times N$ \emph{grid} is the graph $G$ of the form $V(G) = \{(i, j) : i \in [M], j \in [N]\}$, and $E(G) = \{(i, j)(i^{\prime}, j^{\prime}) : |i - i^{\prime}| + |j - j^{\prime}| = 1, i, i^{\prime} \in [M], j, j^{\prime} \in [N]\}$.

{Let $X$ and $Y$ be multisets of vertices of a graph $G$ (i.e., $X$ and $Y$ can contain several copies of the same vertex).}
{We say that $X$ and $Y$ of the same size are \emph{adjacent} if there is a bijective mapping $\alpha: X \rightarrow Y$ such that for $x \in X$, either $x = \alpha(x)$ or $x$ and $\alpha(x)$ are adjacent in $G$.}

\paragraph*{Parameterized complexity}
\label{prelim:pc}

An instance of a parameterized problem $\Pi$ {consists} of an input $I$, which is an input of the non-parameterized version of the problem, and an integer $k$, which is called the \emph{parameter}.
A problem $\Pi$ is said to be \emph{fixed-parameter tractable}, or \FPT, if given an instance $(I,k)$ of $\Pi$, we can decide whether  $(I,k)$ is a \yes-instance of $\Pi$ in  time $f(k)\cdot |I|^{\calO(1)}$.
Here, {$f: \mathbb{N} \mapsto \mathbb{N}$} is some computable function {depending} only on $k$.
Parameterized complexity theory  provides tools to {rule out} the existence of \FPT\ algorithms under plausible complexity-theoretic assumptions.
For this, a hierarchy of parameterized complexity classes $\FPT \subseteq \W[1]\subseteq \W[2] \cdots \subseteq \XP$ was introduced, and it was conjectured that the inclusions are proper.
The most common way to show that it is unlikely that a parameterized problem {admits} an \FPT\ algorithm is to show that it is $\W[1]$ or $\W[2]$-\hard.
It is possible to use reductions analogous to the polynomial-time reductions employed in classical complexity.
Here, the concept of \W$[1]$-\hard ness replaces the one of \NP-\hard ness, and we need not only {to} construct an equivalent instance  \FPT\ time, but also {to} ensure that the size of the parameter in the new instance depends only on the size of the parameter in the original instance.
These types of reductions are called \emph{parameter preserving reductions}.
For {a} detailed introduction to parameterized complexity and related terminologies, we refer the reader to the recent book by Cygan et al.~\cite{DBLP:books/sp/CyganFKLMPPS15}.

A \emph{reduction rule} is a polynomial-time  algorithm that takes as input an instance of a problem and outputs another, usually reduced, instance.
A reduction rule said to be \emph{applicable} on an instance if the output {instance} and input instance are different.
A reduction rule is \emph{safe} if the input instance is a \yes-instance if and only if the output instance is a \yes-instance.

A \emph{kernelization} of a parameterized problem $\Pi_1$ is a polynomial algorithm that maps each instance $(I_1, k_1)$ of $\Pi_1$ to an instance $I$ of $\Pi_2$ such that
$(1)$ $(I, k)$ is a \yes-instance of $\Pi_1$ if and only if $I_2$ is a \yes-instance of $\Pi_2$, and
$(2)$ the size of  $I_2$ is bounded by $g(k)$ for a computable function $g(\cdot)$.
We say a compression is a \emph{polynomial compression}
If $g(\cdot)$ is a polynomial function, then we call it a \emph{polynomial kernel}.
It is known that a problem is \FPT\ if and only if it admits a kernel (See, for example, \cite[Lemma 2.2]{DBLP:books/sp/CyganFKLMPPS15}).

\paragraph*{Rendezvous Games with Adversaries}
\label{prelim:rendezvous}

{Recall that the game is played on a connected graph $G$, and $s$ and $t$ are initial positions of the agents of Facilitator.}
{Let also $k$ be the number of agents of Divider.}

{Notice that a placement of the agents of Facilitator is defined by a multiset of two vertices, as $R$ and $J$ can occupy the same vertex.}
{We denote by $\calF_G$ the family of all multisets of two vertices.}
{Similarly, a placement of $k$ agents of Divider is defined by a multiset of $k$ vertices, because several agents can occupy the same vertex.}
{Let $\calD_G^k$ be the family of all multisets of $k$ vertices.}
{We say that $F \in \calF_G$ and $D \in \calD_G^k$ are \emph{compatible} if $F \cap D = \phi$.}
{Notice that the number of pairs of compatible $F \in \calF_G$ and $D \in \calD_G^k$ is $n{n+k-2 \choose k} + {n \choose 2}{n+k-3 \choose k}$.}
{We denote by}
$$
{\calP_G^k = \{(F, D)|F \in \calF_G, D \in \calD_G^k \text{ s.t. } F \text{ and } D \text{ are compatible}\}}
$$
{the set of \emph{positions} in the game.}

{Formally, a \emph{strategy} of Facilitator for \textsc{Rendezvous} is a function $f: \calP_G^k \rightarrow \calF_G$ that maps $(F, D) \in \calP_G^k$ to $F^{\prime} \in \calF_G$ such that $F$ and $F^{\prime}$ are adjacent and $F^{\prime}$ is compatible with $D$.}
{In words, given a position $(F, D)$, Facilitator moves $R$ and $J$ from $F$ to $F^{\prime}$ if this is her turn to move.}
{Similarly, a \emph{strategy} of Divider is a function $d: \calP_G^k \rightarrow \calD_G^k$ that maps $(F, D) \in \calP_G^k$ to $D^{\prime} \in \calF_G^k$ such that $D$ and $D^{\prime}$ are adjacent and $D^{\prime}$ is compatible with $F$, that is, Divider moves his agents from $D$ to $D^{\prime}$ if this is his turm to move.}
{To accommodate the initial placement, we extend the definition of $d$ for the pair $(\{s, t\}, \phi)$ and let $d(\{s, t\}, \phi) = D^{\prime}$, where $D^{\prime} \in \calD_G^k$ is compatible with $\{s, t\}$.}

Another variant of the game is when the number of moves of the players is at most some parameter $\tau$.
Then Facilitator wins if $R$ and $J$ meet within the first $\tau$ moves, and Divider wins otherwise.
Thus the problem is:

\defproblemboxed{Rendezvous in Time}{A graph $G$ with two given vertices $s$ and $t$, and positive integers $k$ and $\tau$.}{Can Facilitator win on $G$ starting from $s$ and $t$ in at most $\tau$ steps against Divider with $k$ agents?}

Notice that, in the above problem, $\tau$ is part of the input.
When $\tau$ is a fixed constant, this generates a family of problems, one for each different value of $\tau$ referred as $\tau$-\textsc{Rendezvous in Time} problem. {The definitions of strategies for \textsc{Rendezvous in Time} are more complicated, because the decisions of the players also depend on the number of the current step.}
{A \emph{strategy} of Facilitator for \textsc{Rendezvous} is a family of functions $f_i: \calP_G^k \rightarrow \calF_G$ for $i \in \{1, \ldots, \tau\}$ such that $f_i$ maps $(F, D) \in \calP_G^k$ to $F^{\prime} \in \calF_G$, where $F$ and $F^{\prime}$ are adjacent and $F^{\prime}$ is compatible with $D$.}
{Facilitator uses $f_i$ for the move in the $i$-th step of the game.}
{A \emph{strategy} of Divider is a family of functions $d_i: \calP_G^k \rightarrow \calD_G^k$ for $i \in \{0, \ldots, \tau - 1\}$ such that for $i \in \{1, \ldots, \tau - 1\}$, $d_i$ maps $(F, D) \in \calP_G^k$ to $D^{\prime} \in \calF_G^k$, where $D$ and $D^{\prime}$ are adjancent and $D^{\prime}$ is compatible with $F$, and $d_0$ maps $(\{s, t\}, \phi)$ to $D^{\prime} \in \calD_G^k$ compatible with $\{s, t\}$ (slightly abusing notation we do not define $d_0$ for the elements of $\calP_G^k$).}

%% file: tree-width.tex
\section{co-para-NP-hardness Parameterized by FVS and Pathwidth}
\label{sec:tree-width}

In this section, we prove that \textsc{Rendezvous} is \paraNP-\hard\ when parameterized by the feedback vertex set number and the pathwidth of the input graph.
To do that, we present a parameter preserving reduction from the \textsc{3-Dimensional Matching} problem, which is known to be \NP-\hard~\cite[SP~1]{DBLP:books/fm/GareyJ79}.
For notational convenience, we work with the following definition of the problem.
An input consists of a universe $\calU = \{\alpha, \beta, \gamma\} \times [n]$, a family $\calF = \{A_1, A_2, \ldots, A_m\}$ of subsets of $\calU$ such that for every $j \in [m]$, set $A_j = \{(\alpha, a_1), (\beta, b_1), (\gamma, c_1)\}$ for some $a_1, b_1, c_1 \in [n]$.
The goal is to find a subset $\calF^{\prime} \subseteq \calF$ that covers $\calU$ (and contains exactly $n$ sets).

\paragraph*{Reduction}
The reduction takes as input an instance $(\calU, \calF)$ of \textsc{3-Dimensional Matching} and returns an instance $(G, s, t, k)$ of \textsc{Rendezvous}.
It defines $M = n^2 + m^2$ where $n = |\calU|/3$ and $m = |\calF|$. We construct the graph $G$ as follows: (c.f. \Cref{fig:tw-paraNP-hard-critical-vertices,fig:tw-paraNP-hard-force-gaurd-pos,fig:tw-paraNP-hard-encoding-sets,fig:tw-paraNP-hard-encoding-elements}).

\subparagraph*{The Base Gadget}
It starts by adding special vertices $s$ and $t$ and two more vertices $g_1$ and $g_2$, and makes them  common neighbours of $s$ and $t$.
We use $P[u, v, d]$ to denote a simple path from $u$ to $v$ that contains $d$ many internal vertices.
\begin{itemize}
\item For every $i \in [n]$\footnote{We use $i$ as well as  $a_1, b_1, c_1$ as running variables in set $[n]$.
We reserve later types of variables for the integer part of elements in sets $\calF$.}, it adds the following simple paths:
\begin{itemize}
\item $P[u_{i}^{0}, u_{i}^{m+1}, m]$,
\item $P[s, u_{i}^{0}, m]$, $P[s, u_{i}^{m+1}, m]$, $P[t, u_{i}^{0}, m]$, and $P[t, u_{i}^{m+1}, m]$.
\end{itemize}
\end{itemize}
See Figure~\ref{fig:tw-paraNP-hard-force-gaurd-pos} for an illustration.

\begin{figure}[h]
\begin{center}
\includegraphics[trim=0.5cm 4cm 0.0cm 2cm, scale=0.25]{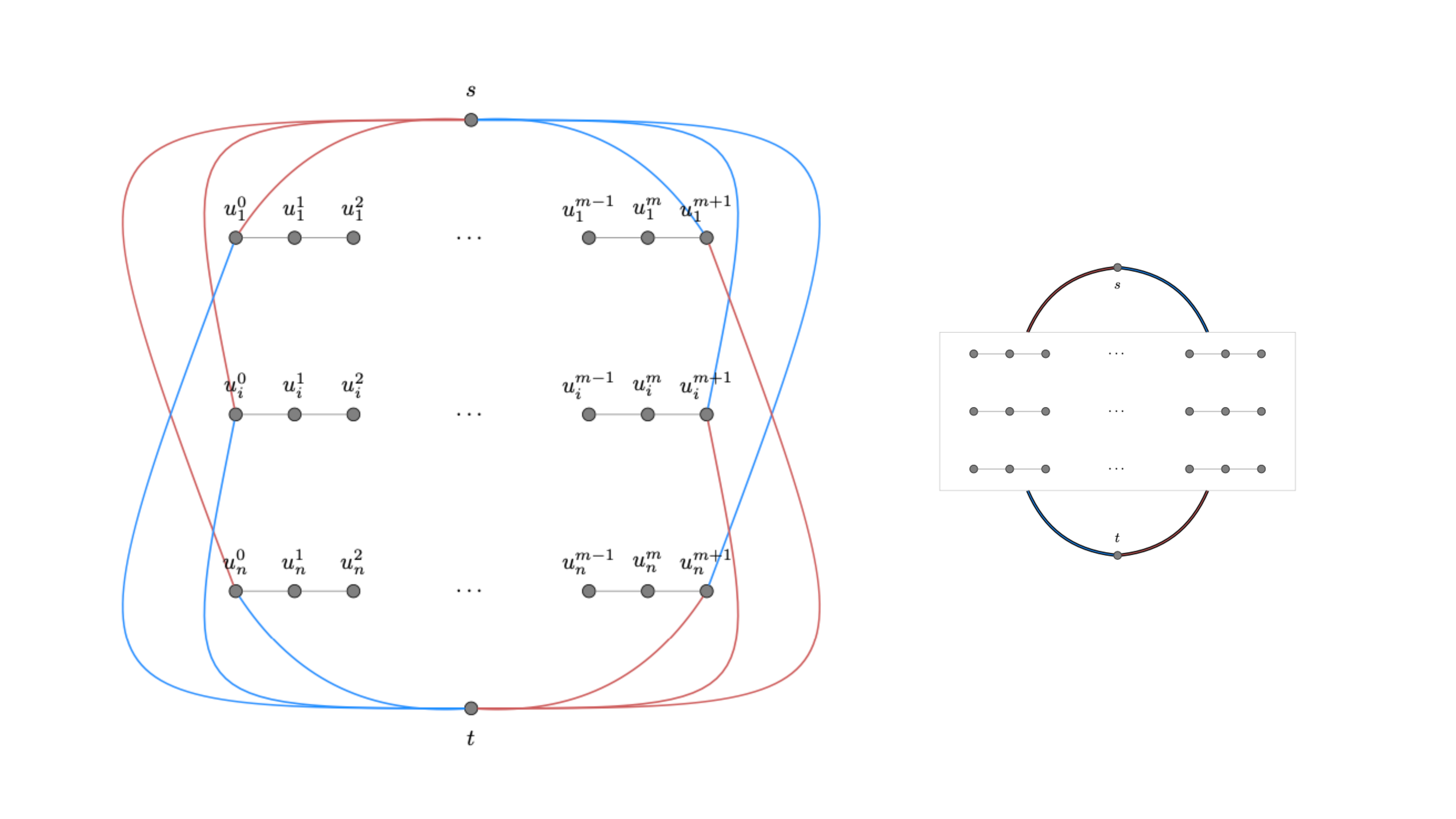}
\end{center}
\caption{
    (Left) The base gadget except the guard vertices $g_1$ and $g_2$ (which are not shown for clarity).
    Each red and blue path has $m$ internal vertices.
    (Right) Schematic representation.
    \label{fig:tw-paraNP-hard-force-gaurd-pos}}
\end{figure}

\subparagraph*{Encoding Elements}
The reduction constructs a symmetric graph to encode elements in $\calU$ and has `left-side' and `right-side.'
It starts by adding vertices $\{\alpha^{\ell}, \beta^{\ell}, \gamma^{\ell}\}$ and $\{\alpha^{r}, \beta^{r}, \gamma^{r}\}$.

\begin{itemize}
\item For every $i \in [n]$, it adds six vertices in $\{\alpha_{i}^{\ell}, \beta_{i}^{\ell}, \gamma_{i}^{\ell}\} \cup \{\alpha_i^{r}, \beta_{i}^{r}, \gamma_{i}^{r}\}$, and the following simple paths:
\begin{itemize}
\item $P[\alpha^{\ell}, \alpha_{i}^{\ell}, M^2 - M \cdot i]$,
$P[\beta^{\ell}, \beta^{\ell}_{i}, M^2 - M \cdot i]$,
$P[\gamma^{\ell}, \gamma_{i}^{\ell}, M^2 - M \cdot i]$,
\item $P[\alpha^{r}, \alpha_{i}^{r}, M^2 + M \cdot i]$,
$P[\beta^{r}, \beta^{r}_{i}, M^2 + M \cdot i]$, and $P[\gamma^{r}, \gamma^{r}_{i}, M^2 + M \cdot i]$.
\end{itemize}
Note that the number of internal vertices in paths from $\alpha^{\ell}$ to $\alpha^{\ell}_i$ and from $\alpha^{r}$ to $\alpha^{r}_i$, and similar such pairs,
are different and depend on $i$.
\item For every $i \in [n]$, it adds six vertices
$\{x_{i}^{\ell}, y_{i}^{\ell}, z_i^{\ell} \} \cup \{x_i^{r}, y_i^{r}, z_i^{r}\}$, and the following simple paths:
\begin{itemize}
\item
$P[x_i^{\ell}, \alpha_i^{\ell}, 2M^2 - 1]$, $P[y_i^{\ell}, \beta_i^{\ell}, 2M^2 - 1]$, $P[z_i^{\ell}, \gamma_i^{\ell}, 2M^2 - 1]$,
\item
$P[x_i^{r}, \alpha_i^{r}, 2M^2 - 1]$, $P[y_i^{r}, \beta_i^{r}, 2M^2 - 1]$, and $P[z_i^{r}, \gamma_i^{r}, 2M^2 - 1]$.
\end{itemize}
\end{itemize}
See Figure~\ref{fig:tw-paraNP-hard-encoding-elements} for an illustration.

\begin{figure}[h]
\begin{center}
\includegraphics[trim=3cm 4cm 0.5cm 3.5cm, scale=0.25]{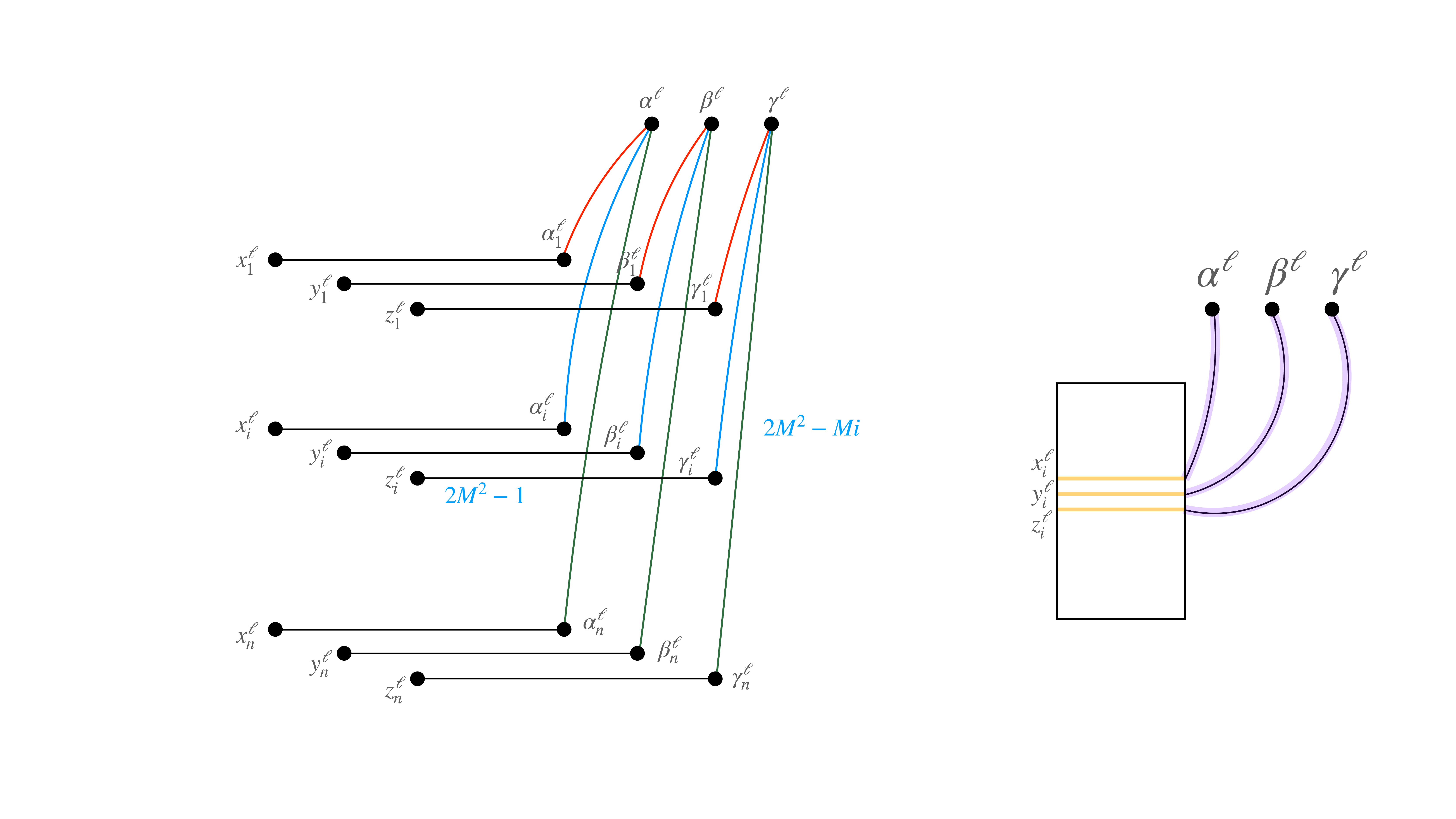}
\end{center}
\caption{
(Left) The left side of the gadget is added to encode elements in $\calU$.
The number of internal vertices in each red, blue, and green path depends on $i$.
The number of internal vertices in each yellow shaded path is $2M^2 - 1$.
(Right) Schematic representation of the gadget.
\label{fig:tw-paraNP-hard-encoding-elements}}
\end{figure}

\subparagraph*{Encoding sets}
The reduction adds simple paths to encode sets.
Consider set $A_j$ for some $j \in [m]$.
Suppose the internal vertices of $P[u_{i}^{0}, u_{i}^{m+1}, m]$ are denoted by $u_{i}^{j}$ for every $j \in [m]$, and $u_{i}^{0}$ is adjacent with $u_{i}^{1}$ and $u_{i}^{m+1}$ is adjacent with $u_{i}^{m}$.
All the vertices in $j^{th}$ 'column' corresponds to set $A_j$.
This, however, is not an encoding of set $A_j$ as it does not provide any information about its elements.
By the definition of the problem, set $A_j$ has an element of the form $(\alpha, a_1)$.
To encode this element, it adds $2n$ many paths connecting $j^{th}$ column to $\alpha^{\ell}$ and to $\alpha^{r}$.
The number of internal vertices in these paths depends on $a_1$.
We encode the remaining two elements in $A_j$ similarly.
We formalise this construction as follows:
\begin{itemize}
\item For every $j \in [m]$, suppose $A_{j} = \{(\alpha, a_1), (\beta, b_1), (\gamma, c_1)\}$.
Then, for every $i \in [n]$, the reduction adds the following six simple paths:
\begin{itemize}
\item $P[\alpha^{\ell}, u_{i}^{j},  M^2 + M \cdot a_1]$,
$P[\beta^{\ell},  u_{i}^{j}, M^2 + M \cdot b_1]$,
 $P[\gamma^{\ell}, u_{i}^{j}, M^2 + M \cdot c_1]$,
\item
$P[\alpha^{r}, u_{i}^{j}, M^2 - M \cdot a_1]$, $P[\beta^{r}, u_{i}^{j}, M^2 - M \cdot b_1]$, and $P[\gamma^{r}, u_{i}^{j}, M^2 - M \cdot c_1]$.
\end{itemize}
\end{itemize}
See Figure~\ref{fig:tw-paraNP-hard-encoding-sets} for an illustration.
\begin{figure}[h]
\begin{center}
\includegraphics[trim=0.5cm 3cm 0.5cm 2.5cm, scale=0.25]{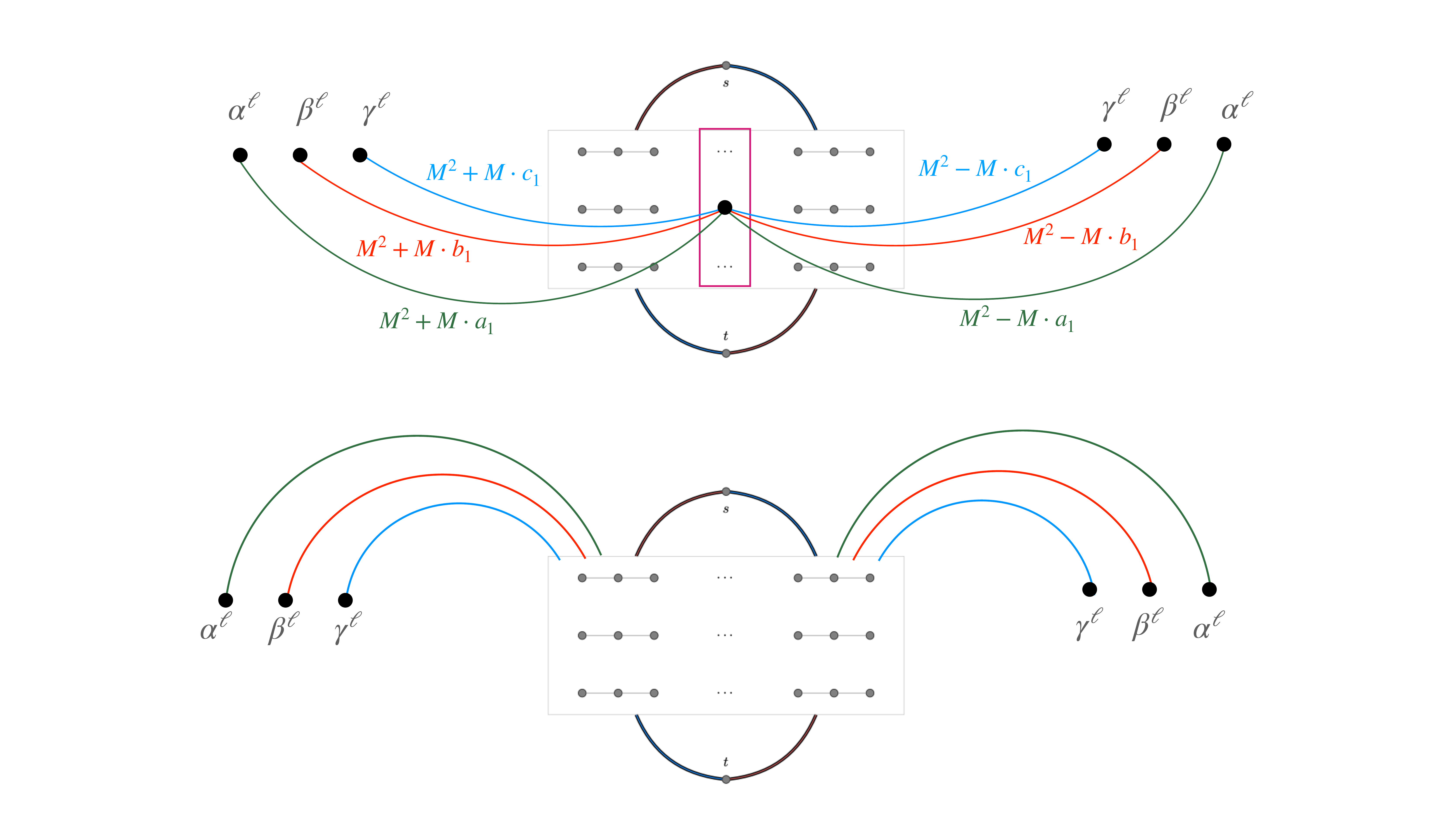}
\end{center}
\caption{
(Top) Vertices added to encode sets in $\calF$.
The number of internal vertices in the paths depend on elements in $A_j$ and are denoted next to it.
(Bottom) Schematic representation of the gadget used in subsequent figures.
\label{fig:tw-paraNP-hard-encoding-sets}}
\end{figure}

\subparagraph*{Critical vertices and connecting paths}
In the last phase of the reduction, it adds critical vertices and connect them to $\{s, t\}$, and also to $x$-type, $y$-type, and $z$-type ends of paths added while encoding elements in $\calU$.
\begin{itemize}
\item For the special vertex $s$, it adds critical vertices, say $s_{\alpha}^{\ell}$, $s_{\beta}^{\ell}$ and $s_{\gamma}^{\ell}$, on the left side.
\begin{itemize}
\item It adds $P[s, s_{\alpha}^{\ell}, 2M^2 + 1]$, $P[s, s_{\beta}^{\ell}, 2M^2 + 1]$, and $P[s, s_{\gamma}^{\ell}, 2M^2 + 1]$.
\item  For every $i \in [n]$, it adds $P[s_{\alpha}^{\ell}, x_{i}^{\ell}, 2M^2]$, $P[s_{\beta}^{\ell}, y_{i}^{\ell}, 2M^2]$, and $P[s_{\gamma}^{\ell}, z_{i}^{\ell}, 2M^2]$.
\end{itemize}
It adds the other critical vertices and paths symmetrically.
We present them for the sake of completeness.

For the special vertex $s$, it adds critical vertices, say $s_{\alpha}^{r}$, $s_{\beta}^{r}$, and $s_{\gamma}^{r}$, on the right side.
\begin{itemize}
\item It adds $P[s, s_{\alpha}^{r}, 2M^2 + 1]$, $P[s, s_{\beta}^{r}, 2M^2 + 1]$, and $P[s, s_{\gamma}^{r}, 2M^2 + 1]$.
\item For every $i \in [n]$, it adds $P[s_{\alpha}^{r}, x_{i}^{r}, 2M^2]$, $P[s_{\beta}^{r}, y_{i}^{r}, 2M^2]$, and $P[s_{\gamma}^{r}, z_{i}^{r}, 2M^2]$.
\end{itemize}

For the special vertex $t$, it adds critical vertices, say $t_{\alpha}^{\ell}$, $t_{\beta}^{\ell}$, $t_{\gamma}^{\ell}$, on the left side.
\begin{itemize}
\item It adds $P[t, t_{\alpha}^{\ell}, 2M^2 + 1]$, $P[t, t_{\beta}^{\ell}, 2M^2 + 1]$, and $P[t + t_{\gamma}^{\ell}, 2M^2 + 1]$.
\item For every $i \in [n]$, it adds $P[t_{\alpha}^{\ell}, x_{i}^{\ell}, 2M^2]$, $P[t_{\beta}^{\ell}, y_{i}^{\ell}, 2M^2]$, $P[t_{\gamma}^{\ell}, z_{i}^{\ell}, 2M^2]$,
\end{itemize}

For the special vertex $t$, it adds critical vertices, say $t_{\alpha}^{r}$, $t_{\beta}^{r}$, and $t_{\gamma}^{r}$, on the right side.
\begin{itemize}
\item It adds $P[t, t_{\alpha}^{r}, 2M^2 + 1]$, $P[t, t_{\beta}^{r}, 2M^2 + 1]$, and $P[t, t_{\gamma}^{r}, 2M^2 + 1]$.
\item For every $i \in [n]$, it adds $P[t_{\alpha}^{r}, x_{i}^{r}, 2M^2]$, $P[t_{\beta}^{r}, y_{i}^{r}, 2M^2]$, and $P[t_{\gamma}^{r}, z_{i}^{r}, 2M^2]$.
\end{itemize}
\end{itemize}

\begin{figure}[h]
\begin{center}
\includegraphics[trim=2.5cm 3cm 0.5cm 2.5cm, scale=0.25]{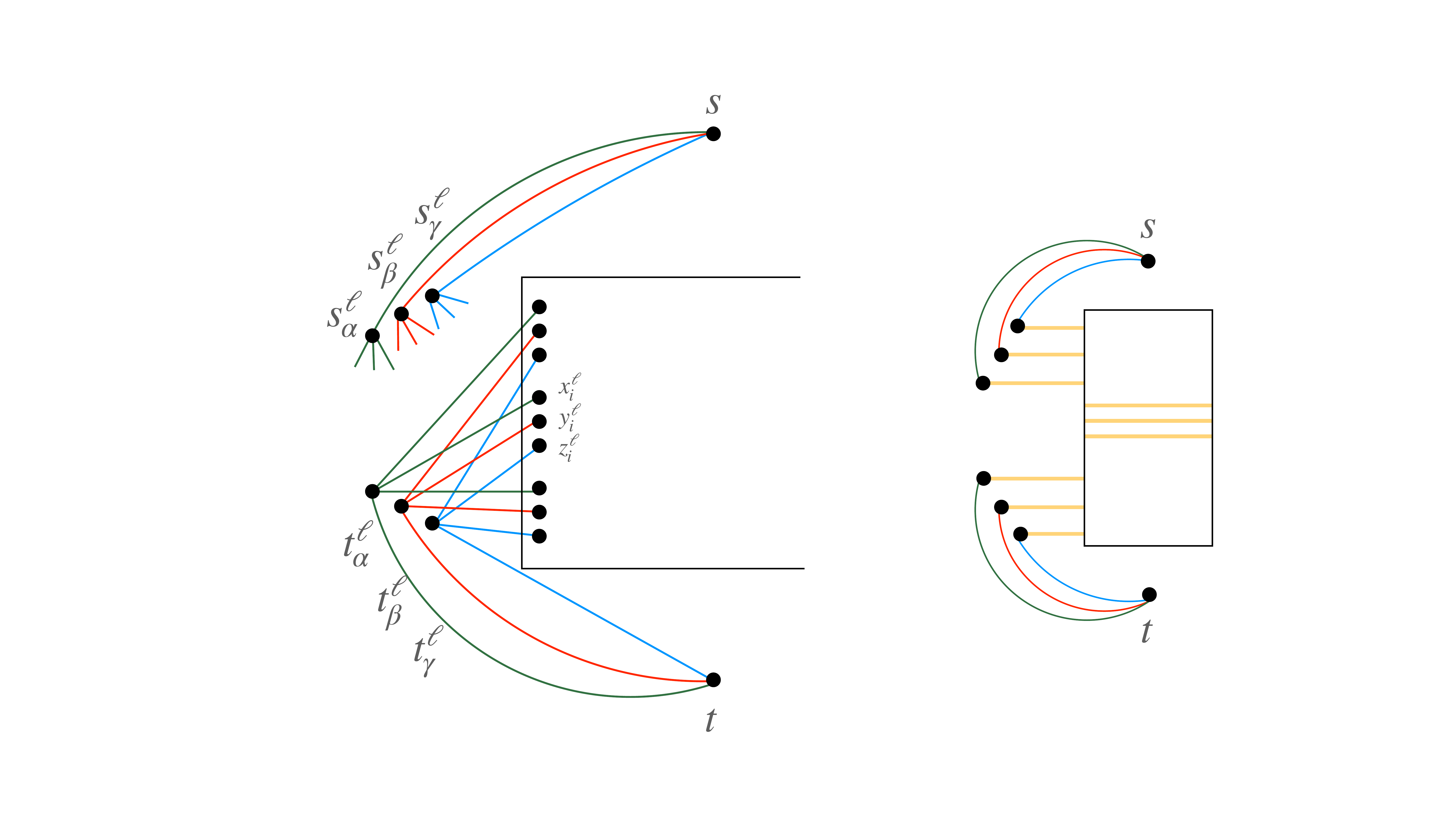}
\end{center}
\caption{
(Left)
Critical vertices added by the reduction.
The number of internal vertices in the paths are fixed ($2M^2 + 1$ or $2M^2$).
(Right) Schematic representation of the gadget.
\label{fig:tw-paraNP-hard-critical-vertices}}
\end{figure}

This completes the construction of the graph $G$.
See Figure~\ref{fig:tw-paraNP-hard-overview} for the overview of the constructed graph.
The reduction sets $k = n + 2$ and returns $(G, s, t, k)$ as the reduced instance of \textsc{Rendezvous}.

\paragraph*{Intuition for the correctness}
We present an intuition for the correctness of the reduction in the reverse direction.
In other words, we state how the initial positions of the Divider's agent correspond to sets and the elements they can cover.
We start with determining the possible initial positions.

As $g_1, g_2$ are common neighbors of  $s$ and $t$, Divider needs to put two agents on $g_1$ and $g_2$.
For the remaining $n$ agents, consider the paths $P[s, u_{i}^{0}, m]$ and $P[u_{i}^{0}, t, m]$ or
$P[s, u_{i}^{m + 1}, m]$ and $P[u_{i}^{m + 1}, t, m]$
for every $i \in [n]$.
Facilitator can move both Romeo and Juliet to $u_i^{0}$ or $u_i^{m+1}$ in $m$ steps.
Hence, Divider needs to place remaining $n$ agents at the positions that are at distance at most  $m$ simultaneously from $u_i^0$ and $u_i^{m+1}$.
We ensure that he needs to place an agent on an internal vertex of $P[u^0_i, u^{m+1}_i, m]$ for every $i \in [n]$.
This will correspond to selecting a set in $\calF$ in a solution.
Formally, an agent at $u^{j}_i$ for some $j \in [m]$ corresponds to selecting $A_j$ in the cover of $\calU$.
As there are $n$ 'rows', this will correspond to selecting $n$ (different) sets from $\calF$.
Hence, the initial position of the Divider's agent will correspond to a collection of sets in $\calF$.

Suppose for every $i \in [n]$, vertices in $\{x^{\ell}_i, x^r_i\}$ correspond to element $(\alpha, i) \in \calU$.
Similarly, vertices in $\{y^{\ell}_i, y^{r}_i\}$ correspond to $(\beta, i)$, and vertices in $\{z^{\ell}_i, z^{r}_{i}\}$ correspond to $(\gamma, i)$.
We say $(\alpha, i)$ is covered if Divider can prevent Facilitator from moving both Romeo and Juliet at $x^{\ell}_i$ as well as at $x^r_i$.

From the Facilitator's preservative, she has $6n$ possible meeting points of the above form.
She can make these choices in two phases.
In the first phase, she can decide to move Romeo towards one of the six vertices in  $\{s^{\ell}_{\alpha}, s^{\ell}_{\beta}, s^{\ell}_{\gamma} \} \cup \{s^{r}_{\alpha}, s^{r}_{\beta}, s^{r}_{\gamma} \}$.
To win, she will have to move Juliet towards the corresponding vertices with respect to $t$.
Suppose she moves Romeo towards $s^{\ell}_{\alpha}$ and Juliet towards $t^{\ell}_{\alpha}$, i.e., Romeo along $P[s, s_{\alpha}^{\ell}, 2M^2 + 1]$ and Juliet along $P[t, t_{\alpha}^{\ell}, 2M^2 + 1]$.
She can move Romeo at $s^{\ell}_{\alpha}$ and Juliet at $t^{\ell}_{\alpha}$ in $2M^2 + 2$ steps.
At this point, she can make one of the $n$ choices and decide to move both Romeo and Juliet towards $x^{\ell}_i$ for some $i \in [n]$, i.e., Romeo along $P[s_{\alpha}^{\ell}, x_{i}^{\ell}, 2M^2]$ and Juliet along $P[t_{\alpha}^{\ell}, x_i^{\ell}, 2M^2]$ for some $i \in [n]$.
See Figure~\ref{fig:tw-paraNP-hard-intuition} for relevant vertices.

\begin{figure}[h]
\begin{center}
\includegraphics[trim=3cm 6.5cm 0.5cm 3cm, scale=0.25]{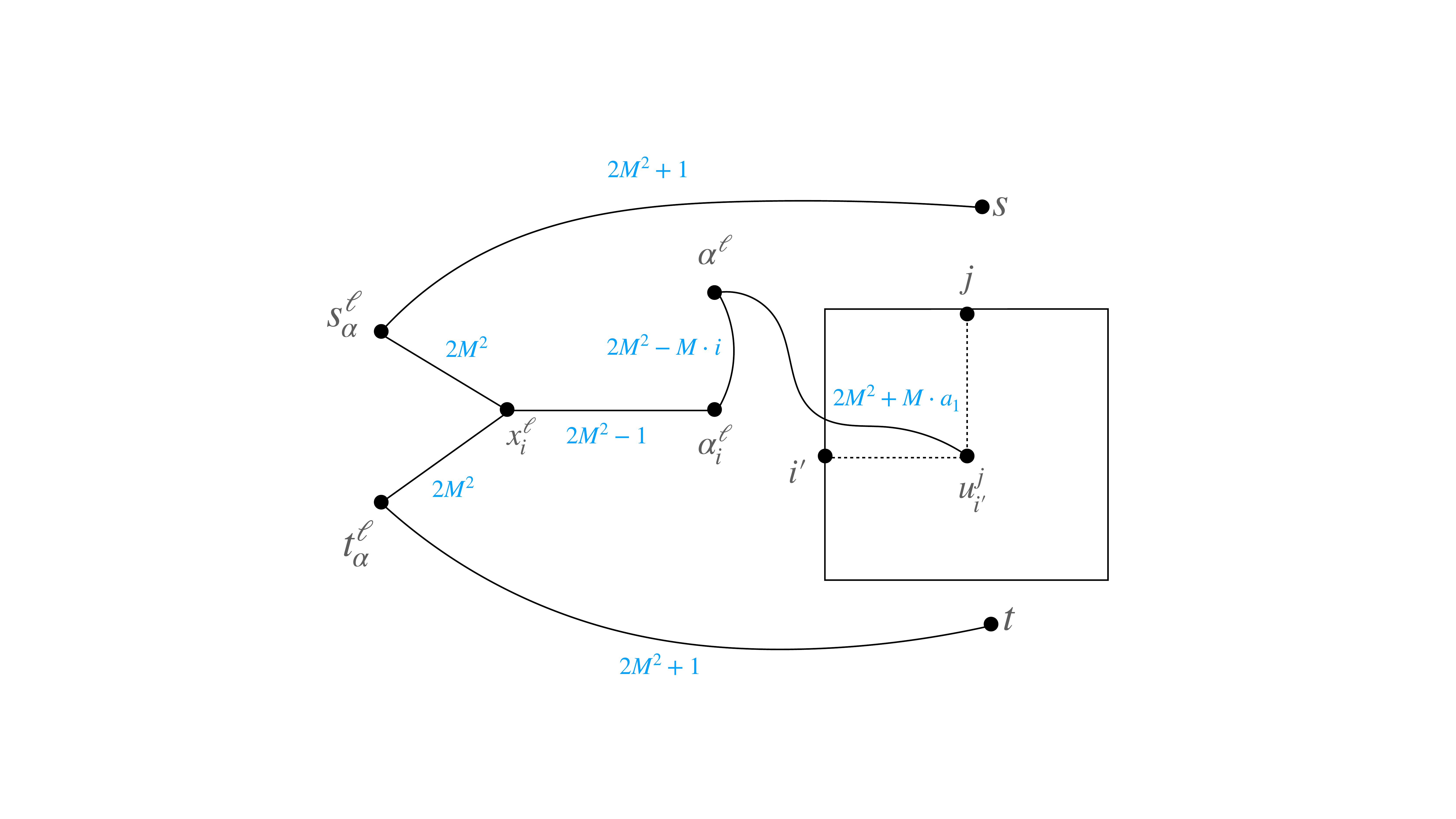}
\end{center}
\caption{
Vertices mentioned while presenting the intuition.
The vertices near the paths indicate the number of internal vertices.
Set $A_j$ contains $(\alpha, a_1)$.
\label{fig:tw-paraNP-hard-intuition}}
\end{figure}

From the Divider's perspective, he can see the first choice made by Facilitator.
However, he has no information about her second choice until next $2M + 2$ steps, i.e., until she moves Romeo at $s^{\ell}_{\alpha}$ and Juliet at $t^{\ell}_{\alpha}$.
Note that Facilitator can move Romeo from $s_{\alpha}^{\ell}$ to $x^{\ell}_i$ and Juliet from $t_{\alpha}^{\ell}$ to $x^{\ell}_i$ in $2M^2 + 1$ steps.
Divider can move an agent from $\alpha^{\ell}_i$ to $x^{\ell}_i$ in $2M^2$ steps.
Considering the initial positions of agents, he needs to ensure that one of its agents is present on $\alpha^{\ell}_{i}$ for \emph{every} $i \in [n]$ in $2M^2 + 2$ steps.
For every $i \in [n]$, he needs to place an agent at $u^{j}_{i'}$ for some $j \in [m]$ and $i' \in [n]$ that he can move it to $\alpha^{\ell}_i$ in $2M^2 + 2$ steps.
We remark that $i$ may not be equal to $i'$.

The only feasible way to do so is by moving the agent from $u^{j}_{i'}$ to $\alpha^{\ell}$ and then move it from $\alpha^{\ell}$ to $\alpha^{\ell}_i$.
Suppose $(\alpha, a_1) \in A_j$ for some $a_1 \in [n]$.
Recall that the number of internal vertices of path from $u^{j}_{i'}$ to $\alpha^{\ell}$ is $M^2 + M \cdot a_1$ where as that of the path from $\alpha^{\ell}$ to $\alpha^{\ell}_i$ is $M^2 - M \cdot i$.
Formally, the number of internal vertices in the path $P[u^j_{i'}, \alpha^{\ell}, M^2 + M\cdot a_1] \circ P[\alpha^{\ell}, \alpha^{\ell}_i, M^2 - M \cdot i]$ is $(M^2 + M \cdot a_1) + 1 + (M^2 - M \cdot i) = 2M^2 + 1 + M^2 \cdot (a_1 - i)$.
This implies that Divider can move an agent from $u^{j}_{i'}$ to $\alpha^{\ell}_i$ in $2M^2 + 2 + M^2 \cdot (a_1 - i)$ steps.
Hence, for \emph{every} $i \in [n]$, Divider should place an agent at $u^{j}_{i'}$ for some $j \in [m]$ and $i' \in [n]$ such that for $(\alpha, a_1) \in A_j$, we have $a_1 \le i$.
Using identical arguments and considering the number of internal vertices on the right side, we prove that for \emph{every} $i \in [n]$, he needs to place an agent at $u^{j'}_{i''}$ for some $j' \in [m]$ and $i'' \in [n]$ such that for $(\alpha, a_2) \in A_j$, we have $a_2 \ge i$.
Combining these two arguments, Divider needs to place an agent at $u^{j}_{i'}$ such that for $(\alpha, a_1) \in A_j$, we have $a_1 = i$.
Moving the agent from this position will prevent Facilitator from moving both Romeo and Juliet at $x^{\ell}_i$ and $x^{r}_i$.
This corresponds to selecting a set from $\calF$ to covers element $(\alpha, i) \in \calU$.
This concludes the intuition for the correctness of the reduction.

Consider set $S := \{s, t\} \cup \{s_{\alpha}^{\ell}, s_{\beta}^{\ell}, s_{\gamma}^{\ell}\} \cup \{\alpha^{\ell}, \beta^{\ell}, \gamma^{\ell}\} \cup \{\alpha^{r}, \beta^{r}, \gamma^{r}\}$ in $G$.
It is easy to verify that $G - S$ is a forest, i.e., the feedback vertex set number of $G$ is at most $14$.
Moreover, every connected component of $G - S$ is either a path or a subdivided caterpillar.
The paths correspond to the paths added while encoding elements in $\calU$ or while adding the critical paths.
The subdivided caterpillars correspond to the base gadgets, and the path added while encoding sets in $\calF^{\prime}$.
Note that the \emph{spine} of the caterpillar is the path $P[u^0_i, u^{m+1}_i, m]$ for some $i \in [n]$ added as a part of base gadget.
This implies that the pathwidth of the resulting graph is at most $16$.

%% file: appendix.tex

We now present arguments formalizing the ideas described above.

\begin{figure}[h!]
\begin{center}
\includegraphics[scale=0.25]{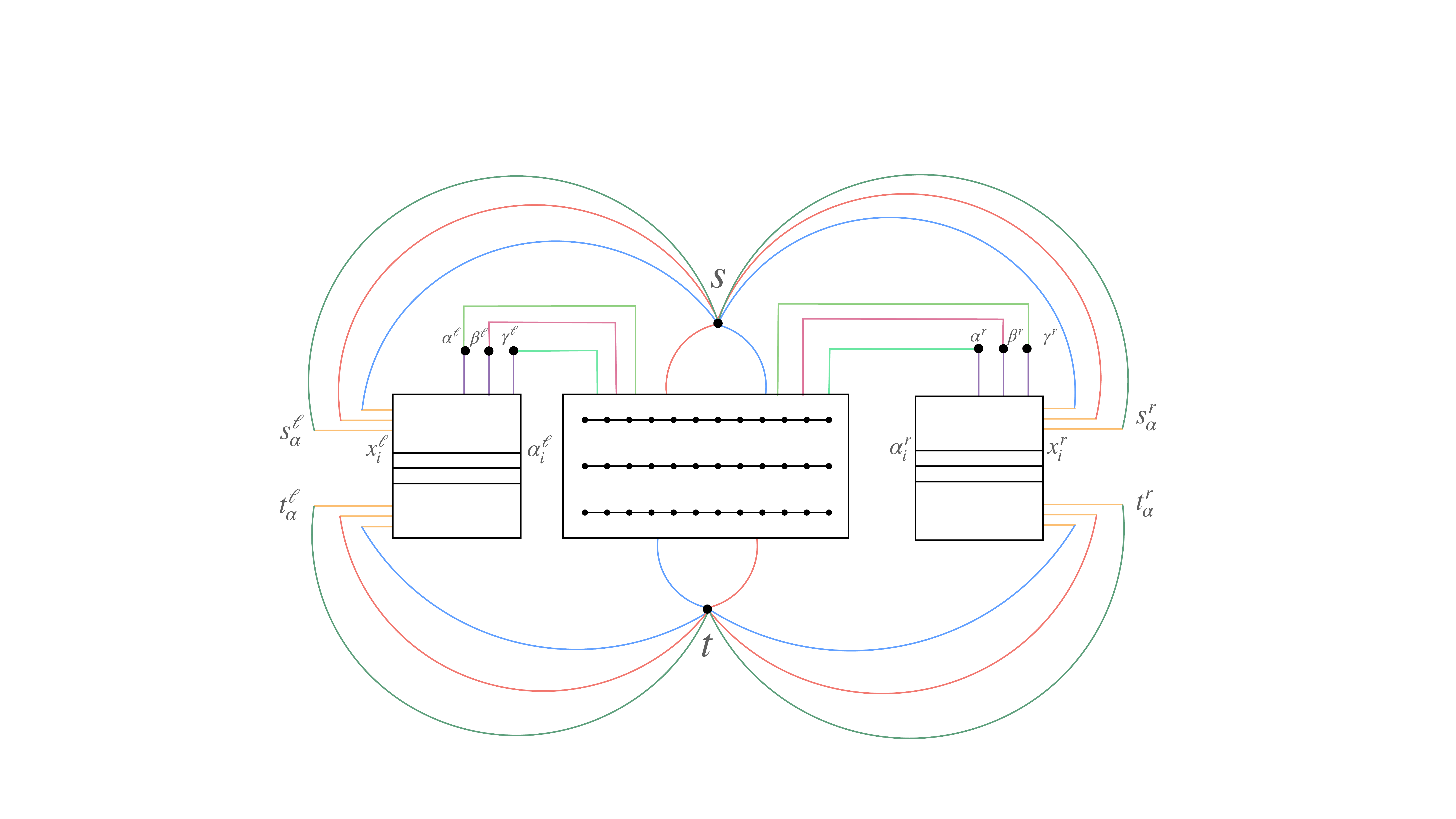}
\end{center}
\caption{
Overview of the reduction in Section~\ref{sec:tree-width}.
\label{fig:tw-paraNP-hard-overview}}
\end{figure}

\begin{lemma}
    \label{lemma:tw-forward-correct}
    If $(\calU, \calF)$ is a \yes-instance of {\sc 3-Dimensional Matching}, then $(G, s, t, n + 2)$ is a \no-instance of {\sc Rendezvous}.
    \end{lemma}
    \begin{proof}
    We show that if $(\calU, \calF)$ is a \yes-instance of \textsc{3-Dimensional Matching}, then Divider with $n + 2$ agents can win in Rendezvous Game with Adversaries.
    Recall that $\calU = \{\alpha, \beta, \gamma\} \times [n]$, and $\calF = \{A_1, A_2, \ldots, A_m\}$ such that $A_j = \{(\alpha, a_1), (\beta, b_1), (\gamma, c_1)\}$ for some $a_1, b_1, c_1 \in [n]$.
    Let $\calF^{\prime} = \{A_{j_1}, A_{j_2}, \ldots, A_{j_n}\} \subseteq \calF$ be the solution for \textsc{3-Dimensional Matching} such that $j_1 < j_2 < \cdots < j_n$.
    Since $\calF^{\prime}$ covers every element of $\calU$, each element of $\calU$ appears in exactly one of the set in $\calF^{\prime}$.

    We describe a {winning strategy} for Divider with the agents $D_1, D_2, \ldots, D_{n+2}$.
    Initially, he puts $D_i$ in the vertex $u_{i}^{j_i}$, for every $i \in [n]$, and $D_{n + 1}$ and $D_{n + 2}$ in $g_1$ and $g_2$ respectively.
    He does not move agents $D_{1}, \ldots, D_{n + 2}$, until Facilitator moves Romeo or Juliet from $s$ or $t$, respectively.
    Suppose Facilitator moves Romeo from $s$ (she may or may not move Juliet from $t$).
    By the construction, she can move Romeo either
    on the paths $P[s, u_{i}^{0}, m]$, $P[s, u_{i}^{m + 1}, m]$ for some $i \in [n]$ or
    on one of the following paths: $P[s, s_{\alpha}^{\ell}, 2M^2 + 1]$, $P[s, s_{\beta}^{\ell}, 2M^2 + 1]$, $P[s, s_{\gamma}^{\ell}, 2M^2 + 1]$, $P[s, s_{\alpha}^{r}, 2M^2 + 1]$, $P[s, s_{\beta}^{r}, 2M^2 + 1]$, or $P[s, s_{\gamma}^{r}, 2M^2 + 1]$.

    Suppose Facilitator moves Romeo from $s$ to a vertex on the path $P[s, u_{i}^{0}, m]$ for some $i \in [n]$.
    Divider moves $D_{n + 1}$ from $g_1$ to $s$ and then towards $u_{i}^{0}$ as she moves Romeo towards $u_{i}^{0}$.
    He also moves $D_i$ to $u_{i}^{0}$ in at most $j_{i}$ steps along the path $P[u_{i}^{0}, u_{i}^{m + 1}, m]$.
    Facilitator needs at least $m + 1$ steps to move
    both Romeo and Juliet in $u_{i}^{0}$ starting from $s$ and $t$ respectively.
    As $j_{i} \le m$, Divider can move $D_{i}$ to $u_{i}^{0}$ before Facilitator can move both Romeo and Juliet to $u_{i}^{0}$.
    Hence, he can block Romeo by $D_i$ and $D_{n + 1}$ on the path $P[s, u_{i}^{0}, m]$.
    Divider keeps moving $D_{i}$ and $D_{n + 1}$ towards Romeo's position and in at most $j_{i} + m - 1$ steps Facilitator can not move Romeo.
    This implies Divider wins by keeping Romeo in its current position with its neighbors occupied by $D_{i}$ and $D_{n + 1}$.
    The argument also follows when Facilitator moves Romeo from $s$ to a vertex on the path $P[s, u_{i}^{m + 1}, m]$ for some $i \in [n]$ since Divider can move $D_{i}$ to $u_{i}^{m + 1}$ in at most $m - j_{i} + 1$ ($\le m$) steps.

    Suppose Facilitator moves Romeo from $s$ to a vertex on the path $P[s, s_{\alpha}^{\ell}, 2M^2 + 1]$.
    In this case, the remaining strategy for Divider is guided by the $\alpha$-type elements in the universe and in sets.
    Recall that each element of $\calU$ appears in exactly one of the set in $\calF^{\prime}$.
    This implies that for every $a_1 \in [n]$, there is a unique set in $\calF'$ that contains element $(\alpha, a_1)$.
    We define function $\psi_{\alpha}: [n] \mapsto [n]$ with respect to $\alpha$.
    More formally, $\psi_{\alpha}(a_1) = i$ if $(\alpha, a_1)$ is contained in $A_{j_i}$ in $\calF'$.
    By the construction, Divider can move agent $D_i$ from $u^{j_i}_i$ to $\alpha_{a_1}^{\ell}$
    in at most $(M^2 - M \cdot a_{1}) + 1 + (M^2 + M \cdot a_{1}) + 1$ steps through the path $P[u_{i}^{j_{i}}, \alpha^{\ell}, M^2 - M \cdot a_{1}] \circ P[\alpha^{\ell}, \alpha_{a_{1}}^{\ell}, M^2 + M \cdot a_{1}]$.
    Hence, Divider can move $D_{i}$ to $\alpha_{a_{1}}^{\ell}$ in at most $2M^2 + 2$ steps.
    Hence, for every $a_1 \in [n]$, Divider can move the agent at $u^{j_i}_i$ to $\alpha_{a_1}^{\ell}$ in $2M^2 + 2$ steps where
    $i = \psi_{\alpha}(a_1)$.
    For notational convenience, we re-write the previous statement while changing the running variable from $a_1$ to $i$.
    Divider can move agents $D_1, D_2, \dots, D_n$ in $2M^2 + 2$ steps such that for every $i \in [n]$, one of its agent is present in $\alpha^{\ell}_i$.
    As in the previous case, he can move $D_{n + 1}$ from $g_1$ to $s$ and then keep moving towards $s_{\alpha}^{\ell}$ as Facilitator moves Romeo towards $s_{\alpha}^{\ell}$.
    He can move $D_{n+2}$ in a similar manner with respect to Juliet.

    After the first move, Facilitator can not move back Romeo and Juliet towards $s$, $t$ respectively, because of the agents $D_{n + 1}$ and $D_{n + 2}$.
    However, she will need at least $2M^2 + 2$ steps to move Romeo to $s_{\alpha}^{\ell}$ and Juliet to $t_{\alpha}^{\ell}$ starting from $s$ and $t$.
    If she moves Romeo to $s_{\alpha}^{\ell}$ and Juliet to $t_{\alpha}^{\ell}$, then she can only move Romeo and Juliet towards $x_{i}^{\ell}$ for some $i \in [n]$.
    However, she will need at least $2M^2 + 1$ steps to move them along the path $P[s_{\alpha}^{\ell}, x_{i}^{\ell}, 2M^2]$ and $P[t_{\alpha}^{\ell}, x_{i}^{\ell}, 2M^2]$, respectively.

    Divider can move his agent from $\alpha^{\ell}_r$ to $x_{i}^{\ell}$ in at most $2M^2$ ($< 2M^2 + 1$) steps along the path $P[\alpha_{i}^{\ell}, x_{i}^{\ell}, 2M^2 - 1]$ for every $i \in [n]$.
    Hence, he can place his agent at $x_{i}^{\ell}$ before Facilitator can move Romeo and Juliet to that place.
    He can keep moving $D_{n + 1}$ towards the position of Romeo and similarly $D_{n + 2}$ towards the position of Juliet.
    Hence, he can block Romeo by the agents at $x^{\ell}_{i}$ for every $i \in [n]$ and $D_{n + 1}$ either on the path $P[s, s_{\alpha}^{\ell}, 2M^2 + 1]$ if Facilitator never moves Romeo at $s^{\ell}_{\alpha}$ or otherwise on the path $P[s_{\alpha}^{\ell}, x_{i}^{\ell}, 2M^2]$ for some $i \in [n]$.

    This implies if Facilitator moves Romeo from $s$ to a vertex on the path $P[s, s^{\ell}_{\alpha}, 2M^2 + 1]$, then Divider has a winning strategy.
    It is easy to see that the similar arguments follows if Facilitator moves Romeo from $s$ to a vertex on the path $P[s, s_{\beta}^{\ell}, 2M^2 + 1]$, $P[s, s_{\gamma}^{\ell}, 2M^2 + 1]$, $P[s, s_{\alpha}^{r}, 2M^2 + 1]$, $P[s, s_{\beta}^{r}, 2M^2 + 1]$, or $P[s, s_{\gamma}^{r}, 2M^2 + 1]$.
    Hence, if $(\calU, \calF)$ is a \yes-instance of \textsc{3-Dimensional Matching}, then Divider with $n + 2$ agents can win in Rendezvous Game with Adversaries, i.e., $(G, s, t, n + 2)$ is a \no-instance of \textsc{Rendezvous}.
    \end{proof}

    \begin{lemma}
    \label{lemma:tw-backward-correct}
    If $(\calU, \calF)$ is a \no-instance of \textsc{3-Dimensional Matching}, then $(G, s, t, n + 2)$ is a \yes-instance of \textsc{Rendezvous}.
    \end{lemma}
\begin{proof}
We show that if $(\calU, \calF)$ is a \no-instance of \textsc{3-Dimensional Matching}, then Facilitator has a winning strategy in at most $4M^2 + 3$ steps against Divider with $n + 2$ agents.

We first consider two simple cases where Facilitator has an easy winning strategy.
First, consider the case when Divider does not place his agents at $g_1$ or $g_2$.
Then, she can move Romeo and Juliet there and win in one step.
Second, consider the case when there is $i \in [n]$ such that none of Divider's agents is within distance $m$ from $u_{i}^{0}$ or from $u_{i}^{m+1}$.
In the first sub-case,
she can move Romeo and Juliet to $u_{i}^{0}$ in $m + 1$ steps through the paths $P[s, u_{i}^{0}, m]$ and $P[t, u_{i}^{0}, m]$, respectively, and win.
Similarly, in the second sub-case she can move Romeo and Juliet to $u_{i}^{m + 1}$ in $m + 1$ steps through the paths $P[s, u_{i}^{m + 1}, m]$ and $P[t, u_{i}^{m + 1}, m]$, respectively, and win.

In the remaining proof, we suppose that Divider places $D_{n + 1}$ at $g_1$ and $D_{n + 2}$ at $g_2$.
Moreover, for every $i \in [n]$, there is a Divider's agent within distance $m$ from $u_{i}^{0}$ and within distance $m$ from $u_{i}^{m + 1}$.
By the construction, the choice of $M$, and the fact that Divider can not place an agent at $s$ or $t$, a single Divider's agent cannot be within distance $m$ from both $u_{i}^{0}$ and $u_{j}^{0}$, or from $u_{i}^{m + 1}$ and $u_{j}^{m + 1}$, or  from $u_{i}^{0}$ and $u_{j}^{m + 1}$, for $i \neq j \in [n]$.
As Divider has $n$ remaining agents, for every $i \in [n]$, there must be an agent, say $D_i$, within distance $m$ from both $u_{i}^{0}$ and $u_{i}^{m + 1}$.
This is possible only when for every $i \in [n]$, $D_i$ is on one of the internal vertices of the path $P[u_{i}^{0}, u_{i}^{m + 1}, m]$ or on one of the paths joining $u_{i}^{j}$ to vertex in $\{\alpha^{\ell}, \beta^{\ell}, \gamma^{\ell}\} \cup \{\alpha^{r}, \beta^{r}, \gamma^{r}\}$, for some $j \in [m]$.
Suppose $\phi: [n] \mapsto [m]$ is the mapping corresponding to the initial position of the Divider's agents.
Formally, for every $i \in [n]$, Divider places agent $D_i$ either on the internal vertex $u_i^{\phi(i)}$ or on the path joining $u_i^{\phi(i)}$ to vertex in $\{\alpha^{\ell}, \beta^{\ell}, \gamma^{\ell}\} \cup \{\alpha^{r}, \beta^{r}, \gamma^{r}\}$.
While defining the mapping, the first condition is prioritized.

We now define the Facilitator's strategy.
Considering $\calU = \{\alpha, \beta, \gamma\} \times [n]$ as the universe, she constructs a collection $\calF$ of subsets of $\calU$ such that for every $j \in [m]$, set $A_j = \{(\alpha, a_1), (\beta, b_1), (\gamma, c_1)\}$ for some $a_1, b_1, c_1 \in [n]$.
Alternately, she reverse-engineers the process used by the reductions to encode sets.
She also constructs a subset $\calF'$ of $\calF$ by considering the initial positions of agents $D_1, D_2, \dots, D_n$.
Formally, she includes $A_{\phi(i)}$ in $\calF^{\prime}$ for every $i \in [n]$, i.e.
$\calF^{\prime} = \{A_{\phi(1)}, A_{\phi(2)}, \ldots, A_{\phi(n)}\}$.
Note that $\calF^{\prime}$ contains exactly $n$ elements\footnote{It is tempting to imagine that if $(\alpha, a_1)$  is not covered by $\calF^{\prime}$ than Facilitator can move Romeo and Juliet at $x^{\ell}_{a_1}$ or at $x^{r}_{a_1}$.
However, we argue that Facilitator can move Romeo and Juliet at $x^{\ell}_{a_2}$ for some $a_2 \le a_1$ or at $x^{r}_{a_3}$ for some $a_3 \ge a_1$.}.

For every $i \in [n]$, define $N({\alpha}, {\le i})$ as the number of sets  $A_{\phi(i')}$ in $\calF^{\prime}$ such that the $\alpha$-element $(\alpha, a_1) \in A_{\phi(i')}$, $a_1 \le i$.
Also, define $N({\alpha}, {\ge i})$ as the number of sets $A_{\phi(i')}$ in $\calF^{\prime}$ such that for the $\alpha$-element $(\alpha, a_1) \in A_{\phi(i')}$, $a_1 \ge i$.
It similarly defines and computes $N(\beta, \le i)$, $N(\beta, \ge i)$, $N(\gamma, \le i)$, and $N(\gamma, \ge i)$.
It then determines whether the following statements are \true.
\begin{enumerate}
\item For all $i \in [n]$, $N(\alpha, \le i) \le i$ and $N(\alpha, \ge i) \le n + 1 - i$.
\item For all $i \in [n]$, $N(\beta, \le i) \le i$ and $N(\beta, \ge i) \le n + 1 - i$.
\item For all $i \in [n]$, $N(\gamma, \le i) \le i$ and $N(\gamma, \ge i) \le n + 1 - i$.
\end{enumerate}
Facilitator has to make two critical choices.
Her first critical choice is at the first step where she has to decide about moving Romeo towards $\{s_{\alpha}^{\ell}, s_{\alpha}^{r}\}$, $\{s_{\beta}^{\ell}, s_{\beta}^{r}\}$, or $\{s_{\gamma}^{\ell}, s_{\gamma}^{r}\}$.
This choice depends on which of the above statement is \emph{false}.
If the first statement is false, she narrows down her choice of moving Romeo either to $s^{\ell}_{\alpha}$ or $s^{\ell}_{\beta}$.
Suppose the first statement is false because of the first inequality for some $i$. In that case, she moves Romeo towards the right side, i.e., towards $s^{r}_{\alpha}$; otherwise, she moves Romeo towards the left side, i.e., towards $s^{\ell}_{\alpha}$.
She moves Juliet towards the corresponding vertex with respect to $t$, i.e., towards $t^{r}_{\alpha}$ and $t^{\ell}_{\alpha}$ in the first and the second case, respectively.

To explain her second choice, suppose, without loss of generality, that the first statement is false because of its first inequality.
She then moves Romeo from $s$ to $s^{r}_{\alpha}$ and Juliet from $t$ to $t^{r}_{\alpha}$ in $2M^2 + 2$ steps.
For her second choice, she finds $i \in [n]$ such that
Divider's agent has not been across $\alpha^{\ell}_i$ since the game started.
She then moves Romeo and Juliet at $x_{i}^{\ell}$ in $2M^2 + 1$ additional moves and wins the game.
She uses a similar strategy in other cases.
We remark that the initial positions of Divider's agents are `close' to the base gadget.
For Divider to move his agent from their initial positions to say vertices like $\alpha^{\ell}_i$ or $\alpha^{r}_i$, he needs to move them via $\alpha^{\ell}$ or $\alpha^{r}$, respectively.

To argue that this is indeed a winning strategy for Facilitator, we first argue that for any initial positions of the Divider's agents, at least one of the three statements above is false.
Suppose $a_1 \in [n]$ is the integer such that $(\alpha, a_1)$ does not appear in any sets in $\calF^{\prime}$.
This implies for every $a^{\prime}_1 \in [a_1]$, element $(\alpha, a^{\prime}_1)$ appears in at least one set in $\calF^{\prime}$.
Suppose every $(\alpha, a^{\prime}_1)$ appears in exactly one set in $\calF^{\prime}$.
As $\calF^{\prime}$ contains $n$ sets,  each set contains an $\alpha$-element, and there $(a_1 - 1)$ sets that contains $\alpha$-element $(\alpha, i)$ such that $i < a_1$, we can conclude the following.
 There are $n - (a_1 - 1)$ sets that contains $\alpha$-element $(\alpha, i)$ such that $i \ge a_1 + 1$.
This implies that the second inequality in the first statement is false for $i = a_1 + 1$.
Consider the case when there is $a^{\prime}_1 \in [a_1]$ such that $(\alpha, a^{\prime}_1)$ appears in at least two sets in $\calF'$.
Suppose $a^{\prime}_1$ is the smallest such integer.
As $a^{\prime}_1 < a_1$, there are at least $a^{\prime}_1 + 1$ many sets in $\calF^{\prime}$ that contains $\alpha$-element $(\alpha, i)$ such that $i \le a^{\prime}_1$.
Hence, in either case, the first statement is false.
Conversely, if the first statement is \true, then $(\alpha, a_1)$ is present in at least one set in $\calF^{\prime}$ for every $a_1 \in [n]$.
This implies if none of the three statements is false, every element in $\calU$ appears in some set in $\calF^{\prime}$.
This, however, contradicts the fact that $(\calU, \calF')$ is a \no-instance.
Hence, for any initial positions of Divider's agents, at least one of the three sentences is false.

This allows Facilitator to make her first choice.
It remains to argue that there exists $i \in [n]$ with desired properties for her second choice.
Towards that, we first identify the conditions in which Divider can move the agent $D_{i^{\prime}}$ to $\alpha^{\ell}_{i}$ in $2M^2 + 2$ steps for two indices $i, i' \in [n]$.
Suppose Divider initially places the agent $D_{i^{\prime}}$ at distance $p_{i^{\prime}}$ from the vertex $u_{i^{\prime}}^{\phi(i^{\prime})}$.
As $D_{i^{\prime}}$ must be within distance $m$ from both $u_{i^{\prime}}^{0}$ and $u_{i^{\prime}}^{m+1}$, we can conclude that $p_{i^{\prime}}$ is at most $m/2$.
Note that $p_{i^{\prime}}$ can be zero.
Suppose $(\alpha, a_1) \in A_{\phi(i^{\prime})}$ for some $a_1 \in [n]$.
Recall that the number of internal vertices of path from $u_{i^{\prime}}^{\phi(i^{\prime})}$ to $\alpha^{\ell}$ is $M^2 + M \cdot a_1$ where as that of the path from $\alpha^{\ell}$ to $\alpha_{i}^{\ell}$ is $M^2 - M \cdot i$.
Hence, Divider can move the agent $D_{i^{\prime}}$ to $\alpha_{i}^{\ell}$ in $2M^2 + 2 + M \cdot (a_1 - i) - p_{i^{\prime}}$ steps if $D_{i^{\prime}}$ is on the path $P[u_{i^{\prime}}^{\phi(i^{\prime})}, \alpha_{\ell}, M^2 + M \cdot a_1]$, and in $2M^2 + 2 + M \cdot (a_1 - i) + p_{i^{\prime}}$ steps otherwise.
Hence, $D_{i^{\prime}}$ can only reach to $\alpha_{i}^{\ell}$ within $2M^2 + 2$ steps if $a_1 \le i$, for any $i \in [n]$, since $p_{i^{\prime}} \ll M$.
Using similar arguments, $D_{i^{\prime}}$ can only reach to $\alpha_{i}^{r}$ within $2M^2 + 2$ steps if $a_1 \ge i$, for any $i \in [n]$.
Note that if $a_1 \le i$ then Divider can not move $D_{i^{\prime}}$ to $\alpha^{r}_{i + 1}$ within $2M^2 + 2$ steps as {$2M^2 + 2 + M \cdot (i + 1 - a_1) \pm p_{i^{\prime}} \ge 2M^2 + 2 + M - m/2 > 2M^2 + 2$}.
Moreover, as the number of internal points between the paths from $\alpha^{r}$ to $\alpha^{r}_i$ is $M^2 + M \cdot i$, if Divider can not move $D_{i^{\prime}}$ to $\alpha^{r}_{i + 1}$ within $2M^2 + 2$ then it can not move it to $\alpha^{r}_{i^{\circ} + 1}$ for any $i^{\circ} \ge i$.

We now argue about the second critical choice of Facilitator.
Suppose the facilitator moves Romeo from $s$ to $s^{r}_{\alpha}$ and moves Juliet from $t$ to $t^{r}_{\alpha}$ according to the strategy.
This implies that the first inequality in the first statement is false for some $i \in [n]$, i.e. $N(\alpha,  \le i) = q > i$.
By the definition of $N(\alpha, \le i)$,  there are $q$ sets in $\calF^{\prime}$ that contains $\alpha$-element $(\alpha, a_1)$ such that $a_1 \le i$.
As discussed in the previous paragraphs,  if $(\alpha, a_1) \in A_{\phi(i^{\prime})}$ for some $i^{\prime} \in [n]$,
then Divider can not move $D_{i^{\prime}}$ from {its initial position} to {$\alpha^{r}_{i+1}$} in $2M^2 + 2$ steps.
This statement is \true\ for $q > i$ many agents.
Consider the set $\{\alpha^{r}_{i + 1}, \alpha^{r}_{i + 2},  \dots,  \alpha^{r}_n\}$ of $n - i$ vertices.
Divider can move at most $n - q\ (< n - i)$ agents to these vertices in $2M^2 + 2$ steps.
Hence, there exists $i^{\circ} \in [n] \setminus [i]$ such that none of the Divider's agent has reached $\alpha^{r}_{i^{\circ}}$.

Suppose for some $i \in [n]$, none of the Divider's agents can reach $\alpha_{i}^{r}$ or any vertex on path $P[x_{i}^{r}, \alpha_{i}^{r}, 2M^2 - 1]$ after $2M^2 + 2$ steps from the start.
Then, every Divider's agent will be at distance at least $2M^2 + 1$ from $x_{i}^{r}$ after $2M^2 + 2$ steps from start.
Note that Facilitator can move Romeo from $s$ to $s_{\alpha}^{r}$ and Juliet from $t$ to $t_{\alpha}^{r}$ in $2M^2 + 2$ steps.
She can then move Romeo and Juliet to $x_{i}^{r}$ in $2M^2 + 1$ steps along the path $P[s_{\alpha}^{r}, x_{i}^{r}, 2M^2]$ and $P[t_{\alpha}^{r}, x_{i}^{r}, 2M^2]$, respectively.
Since Facilitator takes the first turn, she can move Romeo and Juliet to $x_{i}^{r}$ before the Divider's agents and win in $4M^2 + 3$ steps.

A similar argument follows when the second inequality of the first statement is false, then Romeo and Juliet can meet at $x_{i}^{\ell}$ for some $i \in [n]$.
Similarly, when second or third statement is false, then also Romeo and Juliet will be able to meet at $\{y_{i}^{\ell}, y_{i}^{r}\}$ or $\{z_{i}^{\ell}, z_{i}^{r}\}$ for some $i \in [n]$, respectively.
As mentioned earlier, Facilitator will decide the `left' or `right' vertex based on which inequality of the statement is false.

This implies that if $(\calU, \calF)$ is a \no-instance of \textsc{3-Dimensional Matching}, then Facilitator wins in at most $4M^2 + 3$ steps against Divider with $n + 2$ agents, i.e., $(G, s, t, n + 2)$ is a \yes-instance of \textsc{Rendezvous}.
\end{proof}

Consider set $S := \{s, t\} \cup \{s_{\alpha}^{\ell}, s_{\beta}^{\ell}, s_{\gamma}^{\ell}\} \cup \{\alpha^{\ell}, \beta^{\ell}, \gamma^{\ell}\} \cup \{\alpha^{r}, \beta^{r}, \gamma^{r}\}$ in $G$.
It is easy to verify that $G - S$ is a forest, i.e., the feedback vertex set number of $G$ is at most $14$.
Moreover, every connected component of $G - S$ is either a path or a subdivided caterpillar.
The paths correspond to the paths added while encoding elements in $\calU$ or while adding the critical paths.
The subdivided caterpillars correspond to the base gadgets, and the path added while encoding sets in $\calF^{\prime}$.
Note that the \emph{spine} of the caterpillar is the path $P[u^0_i, u^{m+1}_i, m]$ for some $i \in [n]$ added as a part of base gadget.
This implies that the pathwidth of the resulting graph is at most $16$.
Lemma~\ref{lemma:tw-forward-correct}, Lemma~\ref{lemma:tw-forward-correct}, and the fact that the reduction can be completed in the time polynomial in the size of input imply Theorem~\ref{thm:fvs-conp-hard} which we restate here.
\fvsconphard*

%% file: feedback-vertex-set.tex
\section{co-W[1]-hardness Parameterized by FVS, Pathwidth, and the Solution Size}
\label{sec:feedback-vertex-set}

In this section, we prove Theorem~\ref{thm:fvs-w-hard} that states \textsc{Rendezvous} is \co-\W[1]-\hard\ when parameterized by the feedback vertex set number or pathwidth and the solution size.
To do that, we present a parameter preserving reduction from the \textsc{(Monotone) NAE-Integer-3-Sat} problem.
For notational convenience, we work with the following definition of the problem.
An input consists of variables $\calX = \{x_{1}, \ldots, x_{n}\}$ that each take a value in the domain $\calD = \{1, \ldots, d^{\star}\}$ and clauses $\calC = \{C_{1}, \ldots, C_{m}\}$ of the form
$$
\operatorname{NAE}\left(x_{i_{1}} \leq d_{1}, x_{i_{2}} \leq d_{2}, x_{i_{3}} \leq d_{3}\right),
$$
where $d_{1}, d_{2}, d_{3} \in [d^{\star}]$.
Such a clause is satisfied if not all three inequalities are \true\ and not all are \false\ (i.e., they are ``not all equal'').
The goal is to find an assignment of the variables that satisfies all given clauses.
Bringmann et. al. \cite{DBLP:journals/jcss/BringmannHML16} proved that \textsc{(Monotone) NAE-Integer-3-Sat} is \W[1]-\hard\ when parameterized by the number of variables.

\begin{figure}[t]
\begin{center}
\includegraphics[scale=0.25]{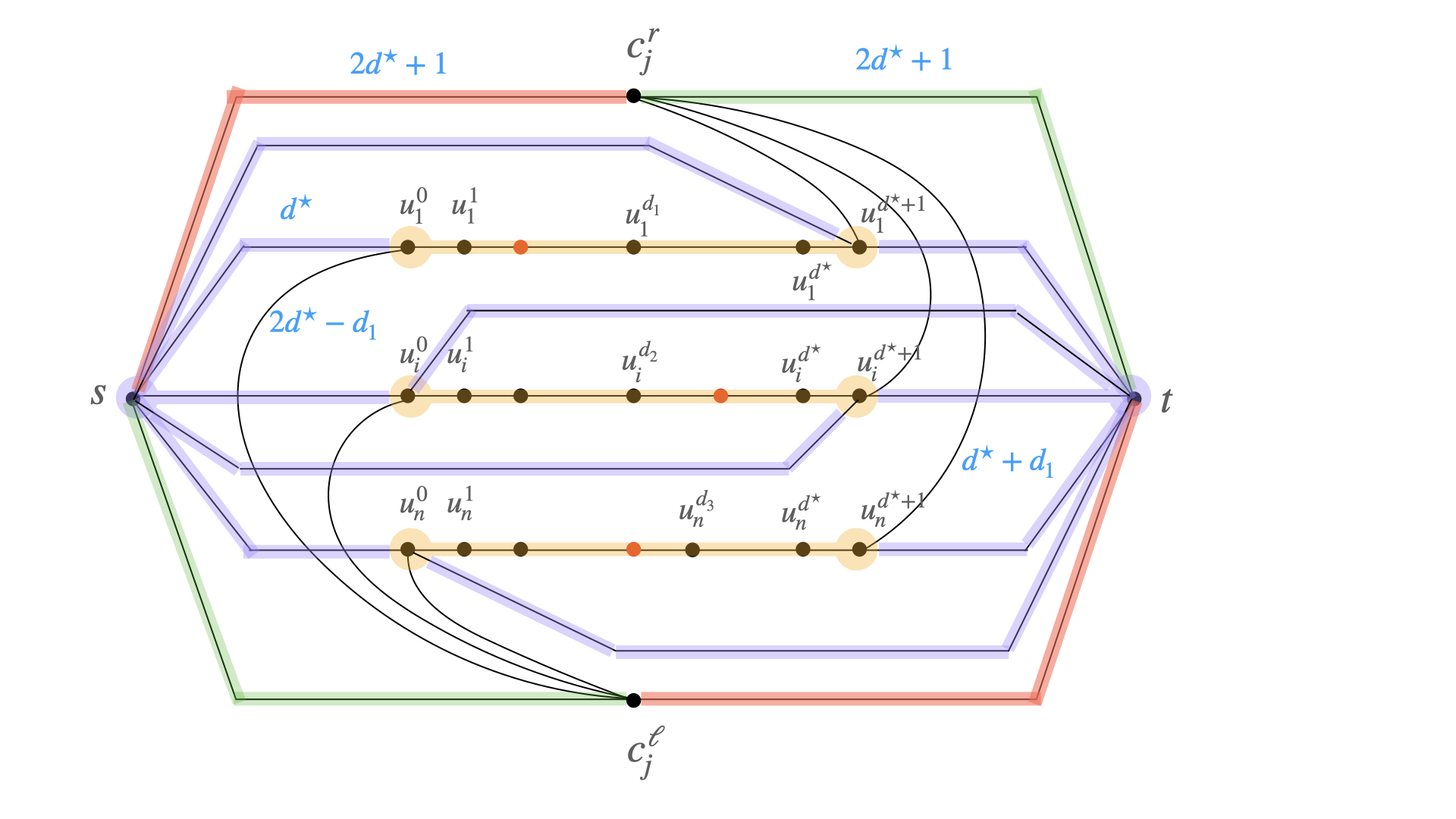}
\end{center}
\caption{
The reduction adds a yellow shaded path for each variable.
Each yellow, purple, or blue shaded path has $d^{\star}$ many internal vertices.
The green and red shaded paths have $2d^{\star} + 1$ many internal vertices.
The number of internal vertices in the remaining path in the figure depends on constants in the clause they are encoding.
Note that vertices $g_1, g_2$ and paths $P[s, u^{d^{\star} + 1}_n, d^{\star}], P[t, u^{0}_1, d^{\star}]$ are not shown in the figure for clarity.
The red vertices denote the positions of the agents.
\label{fig:fvs-w-hard}}
\end{figure}

\paragraph*{Reduction}
The reduction takes as input an instance $(\calX, \calD, \calC)$ of \textsc{(Monotone) NAE-Integer-3-Sat} and returns an instance $(G, s, t, k)$ of \textsc{Rendezvous}.
We construct the graph $G$ as follows: (See Figure~\ref{fig:fvs-w-hard} for the overview of the constructed graph.)

\subparagraph*{The Variable Gadget}
Recall that we use $P[u, v, d]$ to denote a simple path from $u$ to $v$ that contains $d$ many internal vertices.
For every $i \in [n]$, it adds a simple path $P[u_{i}^{0}, u_{i}^{d^{\star} + 1}, d^{\star}]$.
Suppose the internal vertices of $P[u_{i}^{0}, u_{i}^{d^{\star} + 1}, d^{\star}]$ are denoted by $u_{i}^{d}$ for every $d \in [d^{\star}]$, and $u_{i}^{0}$ is adjacent with $u_{i}^{1}$ and $u_{i}^{d^{\star} + 1}$ is adjacent with $u_{i}^{d^{\star}}$.

\subparagraph*{The Clause Gadget}
For every $j \in [m]$, the reduction adds two vertices $c_{j}^{\ell}$ and $c_{j}^{r}$.
Suppose $C_{j} =  \operatorname{NAE}\left(x_{i_{1}} \leq d_{1}, x_{i_{2}} \leq d_{2}, x_{i_{3}} \leq d_{3}\right)$ for some $j \in [m]$.
To encode the inequality $x_{i_1} \leq d_1$, the reduction adds simple paths $P[c_{j}^{\ell}, u_{i_1}^{0}, 2d^{\star} - d_1]$ and $P[c_{j}^{r}, u_{i_1}^{d^{\star} + 1}, d^{\star} + d_1]$.
It encodes the other two inequalities similarly.
We highlight that the number of internal vertices in these simple paths depends on the constant in the inequalities they encode.

\subparagraph*{Critical vertices and connecting paths}
The reduction adds special vertices $s$ and $t$ and two more vertices $g_1$ and $g_2$, and makes them  common neighbours of $s$ and $t$.
\begin{itemize}
\item For every $i \in [n]$, it adds the following simple paths:
\begin{itemize}
\item $P[s, u_{i}^{0}, d^{\star}]$,
 $P[s, u_{i}^{d^{\star} + 1}, d^{\star}]$,
\item $P[t, u_{i}^{0}, d^{\star}]$,
 $P[t, u_{i}^{d^{\star} + 1}, d^{\star}]$.
\end{itemize}
\item For every $j \in [m]$, it adds the following simple paths:
\begin{itemize}
\item $P[s, c_{j}^{\ell}, 2d^{\star} + 1]$,
 $P[s, c_{j}^{r}, 2d^{\star} + 1]$,
\item $P[t, c_{j}^{\ell}, 2d^{\star} + 1]$,
 $P[t, c_{j}^{r}, 2d^{\star} + 1]$.
\end{itemize}
\end{itemize}

This completes the construction of the graph $G$.
The reduction sets $k = n + 2$ and returns $(G, s, t, k)$ as the reduced instance of \textsc{Rendezvous}.

\paragraph*{Intuition for the correctness}
We present an intuition for the correctness of the reduction.
Recall that we use $P[u, v, d_1] \circ P[v, w, d_2]$ to denote the unique path from $u$ to $w$ that contains $v$.
Consider the paths $P[s, u_{i}^{0}, d^{\star}] \circ P[u_{i}^{0}, t, d^{\star}]$ and $P[s, u_{i}^{d^{\star} + 1}, d^{\star}] \circ P[u_{i}^{d^{\star} + 1}, t, d^{\star}]$ for every $i \in [n]$ and paths $P[s, c_{j}^{\ell}, 2d^{\star} + 1] \circ P[c_{j}^{\ell}, t, 2d^{\star} + 1]$ and $P[s, c_{j}^{r}, 2d^{\star} + 1] \circ P[c_{j}^{r}, t, 2d^{\star} + 1]$ for every $j \in [m]$.
As we will see, the only way Facilitator can win in Rendezvous Games with Adversaries is by moving Romeo and Juliet along with one of these $2n + 2m$ paths.
As $g_1, g_2$ are common neighbors of $s$ and $t$, Divider needs to put two of the $k = n + 2$ agents on $g_1$ and $g_2$.
Suppose he puts the remaining $n$ agents at some internal vertices of the paths added while encoding variables.
He places the agents such that each path contains one of them.
For example, the red vertices in Figure~\ref{fig:fvs-w-hard} corresponds to the positions of the agents on the paths added while encoding variables $x_1$, $x_2$, and $x_3$.

Suppose Divider places an agent at an internal vertex, say $u_{1}^{d}$, of $P[u_{1}^{0}, u_{1}^{d^{\star} + 1}, d^{\star}]$.
Facilitator can move Romeo and Juliet to either $u_{1}^{0}$ or $u_{1}^{d^{\star} + 1}$ in $d^{\star} + 1$ steps.
The length of the path from $u_{1}^{0}$ to $u_{1}^{d}$ is $d$ and the length of the path from $u_{1}^{d^{\star} + 1}$ to $u_{1}^{d}$ is $d^{\star} - d + 1$.
\begin{itemize}
\item[-] Divider can move the agent from $u_{1}^{d}$ to $u_{1}^{0}$ in at most $d^{\star}$ steps as $d \le d^{\star}$, and
\item[-] Divider can move the agent from $u_{1}^{d}$ to $u_{1}^{d^{\star} + 1}$ in at most $d^{\star}$ steps as $d^{\star} - d + 1 \le d^{\star}$.
\end{itemize}

Recall that the simple path $P[c_{j}^{\ell}, u_{1}^{0}, 2d^{\star} - d_1]$, as the notation suggests, has $2d^{\star} - d_1$ internal vertices.
Hence, the path $P[c_{j}^{\ell}, u_{1}^{0}, 2d^{\star} - d_1] \circ P[u_{1}^{0}, u_{1}^{d}, d - 1]$ has $(2d^{\star} - d_1) + 1 + (d - 1)$ many internal vertices.
Hence, the length of path from $c_{j}^{\ell}$ to $u_{1}^{d}$ is $2d^{\star} + 1 + d - d_1$.
\begin{itemize}
\item[-] Divider can move the agent from $u_{1}^{d}$ to $c_{j}^{\ell}$ in at most $2d^{\star} + 1$ steps only if $d \le d_1$.
\end{itemize}
Consider symmetric arguments for $c_{j}^{r}$.
The simple path $P[c_{j}^{r}, u_{1}^{d^{\star} + 1}, d^{\star} + d_1]$ has $d^{\star} + d_1$ many internal vertices.
Hence, the path $P[c_{j}^{r}, u_{1}^{d^{\star} + 1}, d^{\star} + d_1] \circ P[u_{1}^{d^{\star} + 1}, u_{1}^{d}, d^{\star} - d]$ has $(d^{\star} + d_1) + 1 + (d^{\star} - d)$ many internal vertices.
Hence, the length of the path from $u_{1}^{d}$ to $c_{j}^{r}$ is $2d^{\star} + 2 + d_1 - d$.
\begin{itemize}
\item[-] Divider can move the agent from $u_{1}^{d}$ to $c_{j}^{\ell}$ in at most $2d^{\star} + 1$ steps only if $d > d_1$.
\end{itemize}

Suppose there is a clause $C_{j} \in \calC$ such that $C_{j} = \text{NAE}(x_1 \le d_1, x_2 \le d_2, x_3 \le d_3)$.
Consider the two vertices $c_{j}^{\ell}$ and $c_{j}^{r}$ added while encoding $C_{j}$.
Note that Facilitator can move Romeo and Juliet to either $c_{j}^{\ell}$ or $c_{j}^{r}$ in $2d^{\star} + 2$ steps.
Moreover, apart from $s$ and $t$, the only branching points in paths $P[s, c_{J}^{\ell}, 2d^{\star} + 1] \circ P[c_{j}^{\ell}, t, 2d^{\star} + 1]$ and $P[s, c_{j}^{r}, 2d^{\star} + 1] \circ P[c_{j}^{r}, t, 2d^{\star} + 1]$ are $c_{j}^{\ell}$ and $c_{j}^{r}$, respectively.
Hence, Divider needs to place an agent that he can move to $c_{j}^{\ell}$ in at most $2d^{\star} + 1$ steps.
Similarly, he needs to place an agent that he can move to $c_{j}^{r}$ in at most $2d^{\star} + 1$ steps.
As we will see, Divider can only move the agents stationed at the paths corresponding to variables $x_1$, $x_2$, or $x_3$ to $c_{j}^{\ell}$ or $c_{j}^{r}$ in at most $2d^{\star} + 1$ steps.
Hence, he needs to place agents at the interior vertices, say $u_{1}^{c_1}$, $u_{2}^{c_2}$, $u_{3}^{c_3}$, of $P[u_{1}^{0}, u_{1}^{d^{\star} + 1}, d^{\star}]$, $P[u_{2}^{0}, u_{2}^{d^{\star} + 1}, d^{\star}]$ and $P[u_{3}^{0}, u_{3}^{d^{\star} + 1}, d^{\star}]$, respectively, such that
\begin{itemize}
\item[-] at least one of the inequalities in $\{c_1 \le d_1; c_2 \le d_2; c_3 \le d_3\}$ is \true, and
\item[-] \emph{simultaneously} at least one of the inequalities in $\{c_1 > d_1; c_2 > d_2; c_3 > d_3\}$ is \true.
\end{itemize}
This position of agents corresponds to the value of variables $x_1, x_2, x_3$ in $[d^{\star}]$ that satisfy the clause $C_{j} = \text{NAE}(x_1 \le d_1, x_2 \le d_2, x_3 \le d_3)$.
In the following two lemmas, we formalize these intuitions.

\begin{lemma}
\label{lemma:fvs-forward-correct}
If $(\calX, \calD, \calC)$ is a \yes-instance of {\sc (Monotone) NAE-Integer-$3$-Sat}, then $(G, s, t, n + 2)$ is a \no-instance of {\sc Rendezvous}.
\end{lemma}
\begin{proof}
We show that if $(\calX, \calD, \calC)$ is a \yes-instance of \textsc{(Monotone) NAE-Integer-$3$-Sat}, then Divider with $n + 2$ agents can win in Rendezvous Game with Adversaries.
Recall that $n = |\calX|$, and $m = |\calC|$.
Suppose $\psi: \calX \rightarrow [d^{\star}]$ be a satisfying assignment, and $\psi(x_i) = d_i$ for every $i \in [n]$.

We describe a {winning strategy} for Dividers with the agents $D_1, D_2, \ldots, D_{n+2}$.
Initially, he puts $D_i$ in the vertex $u_{i}^{d_i}$, for every $i \in [n]$, and $D_{n + 1}$ and $D_{n + 2}$ in $g_1$ and $g_2$ respectively.
He does not move agents $D_{1}, \ldots, D_{n + 2}$, until Facilitator moves Romeo or Juliet from $s$ or $t$, respectively.
Suppose without loss of generality Facilitator first moves Romeo from $s$ (she may or may not move Juliet from $t$).
By the construction, she can move Romeo either
on the paths $P[s, u_{i}^{0}, d^{\star}]$, $P[s, u_{i}^{d^{\star} + 1}, d^{\star}]$ for some $i \in [n]$ or
on the paths $P[s, c_{j}^{\ell}, 2d^{\star} + 1]$, $P[s, c_{j}^{r}, 2d^{\star} + 1]$ for some $j \in [m]$.


Suppose Facilitator moves Romeo from $s$ to a vertex on the path $P[s, u_{i}^{0}, d^{\star}]$ for some $i \in [n]$.
Divider moves $D_{n + 1}$ from $g_1$ to $s$ and then towards $u_{i}^{0}$ as she moves Romeo towards $u_{i}^{0}$.
He also moves $D_i$ to $u_{i}^{0}$ in at most $\psi(x_i)$ steps along the path $P[u_{i}^{0}, u_{i}^{d^{\star} + 1}, d^{\star}]$.
Facilitator needs at least $d^{\star} + 1$ steps to move
both Romeo and Juliet in $u_{i}^{0}$ starting from $s$ and $t$ respectively.
As $\psi(x_i) \le d^{\star}$, Divider can move $D_{i}$ to $u_{i}^{0}$ before Facilitator can move both Romeo and Juliet to $u_{i}^{0}$.
Hence, he can block Romeo by $D_i$ and $D_{n + 1}$ on the path $P[s, u_{i}^{0}, d^{\star}]$.
Divider keeps moving $D_{i}$ and $D_{n + 1}$ towards Romeo's position and in at most $\psi(x_i) + d^{\star} - 1$ steps Facilitator can not move Romeo.
This implies Divider wins by keeping Romeo in its current position with its neighbors occupied by $D_{i}$ and $D_{n + 1}$.
The argument also follows when Facilitator moves Romeo from $s$ to a vertex on the path $P[s, u_{i}^{d^{\star} + 1}, d^{\star}]$ for some $i \in [n]$ since Divider can move $D_{i}$ to $u_{i}^{d^{\star} + 1}$ in at most $d^{\star} - \psi(x_i) + 1$ ($\le d^{\star}$) steps.

Suppose Facilitator moves Romeo from $s$ to a vertex on the path $P[s, c_{j}^{\ell}, 2d^{\star} + 1]$ for some $j \in [m]$.
Let $C_{j} = \text{NAE}(x_{i_1} \le d_1, x_{i_2} \le d_2, x_{i_3} \le d_3)$.
Since $\psi$ is a satisfying assignment, it sets the values of variables such that at least one of the inequalities will be \true\ and at least one of the inequalities will be \false.
We assume without loss of generality that $\psi(x_{i_1}) \le d_1$ and $\psi(x_{i_2}) > d_2$.
Divider moves $D_{i_1}$ to $c_{j}^{\ell}$ in at most $2d^{\star} - d_1 + 1 + \psi(x_{i_1})$ steps through the path $P[c_{j}^{\ell}, u_{i_1}^{0}, 2d^{\star} - d_1] \circ P[u_{i_1}^{0}, u_{i_1}^{d^{\star} + 1}, d^{\star}]$.
As in the previous case, he can move $D_{n+1}$ from $g_1$ to $s$ and then keep moving towards $c_{j}^{\ell}$ as Facilitator moves Romeo towards $c_{j}^{\ell}$.
He can move $D_{n+2}$ in a similar manner with respect to Juliet.

Facilitator can move both Romeo and Juliet to $c_{j}^{\ell}$ in at least $2d^{\star} + 2$ steps starting from $s$ and $t$ respectively.
Divider can move $D_{i_1}$ to $c_{j}^{\ell}$ before Romeo and Juliet as $2d^{\star} - d_1 + 1 + \psi(x_{i_1}) \le 2d^{\star} + 1$.
Hence, Romeo is blocked by $D_{i_1}$ and $D_{n+1}$ on the path $P[s, c_{j}^{\ell}, 2d^{\star} + 1]$ and Juliet cannot reach Romeo.
Divider keeps moving $D_{i_1}$ and $D_{n+1}$ towards Romeo and in at most $4d^{\star} - d_1 + 1 + \psi(x_{i_1})$ steps Romeo cannot move.
This implies Divider wins.
The argument also follows when Facilitator moves Romeo from $s$ to a vertex on the path $P[s, c_{j}^{r}, 2d^{\star} + 1]$ for some $j \in [m]$ since Divider can move $D_{i_2}$ to $c_{j}^{r}$ in at most $2d^{\star} + 2 + d_2 - \psi(x_{i_2})$ ($ < 2d^{\star} + 2$) steps.

This implies that if $(\calX, \calD, \calC)$ is a \yes-instance of \textsc{(Monotone) NAE-Integer-$3$-Sat}, then Divider with $n + 2$ agents can win in Rendezvous Game with Adversaries, i.e., $(G, s, t, n + 2)$ is a \no-instance of \textsc{Rendezvous}.
\end{proof}

\begin{lemma}
\label{lemma:fvs-backward-correct}
If $(\calX, \calD, \calC)$ is a \no-instance of {\sc (Monotone) NAE-Integer-$3$-Sat}, then $(G, s, t, n + 2)$ is a \yes-instance of {\sc Rendezvous}.
\end{lemma}
\begin{proof}
We show that if $(\calX, \calD, \calC)$ is a \no-instance of \textsc{(Monotone) NAE-Integer-$3$-Sat}, then Facilitator wins in at most $2d^{\star} + 2$ steps against Divider with $n + 2$ agents.

We first consider two simple cases where Facilitator has an easy winning strategy.
First, consider the case when Divider does not place his agents at $g_1$ or $g_2$.
Then, she can move Romeo and Juliet there and win in one step.
Second, consider the case when there is $i \in [n]$ such that none of Divider's agents is within distance $d^{\star}$ from $u_{i}^{0}$ or from $u_{i}^{d^{\star}+1}$.
In the first sub-case,
she can move Romeo and Juliet to $u_{i}^{0}$ in $d^{\star} + 1$ steps through the paths $P[s, u_{i}^{0}, d^{\star}]$ and $P[t, u_{i}^{0}, d^{\star}]$, respectively, and win.
Similarly, in the second sub-case she can move Romeo and Juliet to $u_{i}^{d^{\star} + 1}$ in $d^{\star} + 1$ steps through the paths $P[s, u_{i}^{d^{\star} + 1}, d^{\star}]$ and $P[t, u_{i}^{d^{\star} + 1}, d^{\star}]$, respectively, and win.

In the remaining proof, we suppose that Divider places $D_{n + 1}$ at $g_1$ and $D_{n + 2}$ at $g_2$.
Moreover, for every $i \in [n]$, there is a Divider's agent within distance $d^{\star}$ from $u_{i}^{0}$ and within distance $d^{\star}$ from $u_{i}^{d^{\star} + 1}$.
Suppose from now that for every $i \in [n]$, there exists a Divider's agent within distance $d^{\star}$ from $u_{i}^{0}$ and within distance $d^{\star}$ from $u_{i}^{d^{\star} + 1}$.
By the construction and the fact that Divider can not place an agent at $s$ or $t$, a single Divider's agent cannot be within distance $d^{\star}$ from both $u_{i}^{0}$ and $u_{j}^{0}$, or $u_{i}^{d^{\star} + 1}$ and $u_{j}^{d^{\star} + 1}$, or $u_{i}^{0}$ and $u_{j}^{d^{\star} + 1}$, for $i \neq j \in [n]$.
As Divider has $n$ remaining agents, for every $i \in [n]$, there must be an agent, say $D_i$, within distance $d^{\star}$ from both $u_{i}^{0}$ and $u_{i}^{d^{\star} + 1}$.
This is possible only when for every $i \in [n]$, $D_i$ is on one of the internal vertices of the path $P[u_{i}^{0}, u_{i}^{d^{\star} + 1}, d^{\star}]$.
Suppose $\phi: [n] \rightarrow [d^{\star}]$ is the mapping corresponding to the initial position of the Divider's agents.
Formally, for every $i \in [n]$, Divider places agent $D_i$ on $u_{i}^{\phi(i)}$.
For every $i \in [n]$, the initial position of $D_i$ also represents a possible assignment of variable $x_i$ in $(\calX, \calD, \calC)$.

We now define the Facilitator's strategy.
Considering $\calX = \{x_1, \ldots, x_n\}$ as the variables that each take a value in the domain $\calD = \{1, \ldots, d^{\star}\}$, she constructs a collection $\calC$ of clauses such that for every $j \in [m]$,
clause $C_{j} = \text{NAE}(x_{i_1} \le d_1, x_{i_2} \le d_2, x_{i_3} \le d_3)$, where $x_{i_1}, x_{i_2}, x_{i_3} \in \calX$ for some $d_1, d_2, d_3 \in [d^{\star}]$.
Alternately, she reverse-engineers the process used by the reductions to encode clauses.
She also constructs an assignment $\psi: \calX \rightarrow \calD = [d^{\star}]$ by considering the initial positions of agents $D_1, D_2, \dots, D_n$.
Formally, $\psi(x_i) = \phi(i)$ for every $i \in [n]$.
It then determines whether the following statements are \true.
\begin{enumerate}
\item For some clause $C_{j} = \text{NAE}(x_{i_1} \le d_1, x_{i_2} \le d_2, x_{i_3} \le d_3)$, all of the inequalities in $\{\psi(x_{i_1}) \le d_1; \psi(x_{i_2}) \le d_2; \psi(x_{i_3}) \le d_3\}$ are \true, where $j \in [m]$.
\item For some clause $C_{j} = \text{NAE}(x_{i_1} \le d_1, x_{i_2} \le d_2, x_{i_3} \le d_3)$, all of the inequalities in $\{\psi(x_{i_1}) \le d_1; \psi(x_{i_2}) \le d_2; \psi(x_{i_3}) \le d_3\}$ are \false, where $j \in [m]$.
\end{enumerate}
Facilitator has to make a critical choice in the first step where she has to decide about moving Romeo towards $c_{1}^{\ell}, \ldots, c_{m}^{\ell}, c_{1}^{r}, \ldots,$ or $c_{m}^{r}$.
This choice depends on which of the above statement is \true\ and for which clause it is \true.
If the first statement is \true\ for the clause $C_{j} \in \calC$, then she moves Romeo and Juliet towards $c_{j}^{r}$.
Similarly, if the second statement is \true\ for the clause $C_{j} \in \calC$, then she moves Romeo and Juliet towards $c_{j}^{\ell}$.

To argue that this is indeed a winning strategy for Facilitator, we first argue that for any initial positions of Divider's agents, at least one of the two statements above is \true.
Assume the above two statements are \false\ for all $j \in [m]$, which implies in all the clauses $C_{j} \in \calC$, not all three inequalities are \true\ and not all are \false\ .
Hence, all the clauses are satisfied by the assignment $\psi$.
This, however, contradicts the fact that $(\calX, \calD, \calC)$ is a \no-instance.
Hence, for any initial positions of Divider's agents, at least one of the two sentences is \true.

This allows Facilitator to make her choice.
It remains to argue that Romeo and Juliet can meet at the vertex $c_{j}^{\ell}$ or $c_{j}^{r}$ which Facilitator has chosen.
Suppose, one of the statements is \true\ for the clause $C_{j}$.
For notational convenience, suppose $C_{j} = \text{NAE}(x_{1} \le d_1, x_{2} \le d_2, x_{3} \le d_3)$.

Suppose $\psi(x_1) \le d_1, \psi(x_2) \le d_2, \psi(x_3) \le d_3$ (i.e. First statement is \true).
Then, as mentioned in the Facilitator's strategy, her choice will be to move Romeo and Juliet towards $c_{j}^{r}$.
For $i \in \{1, 2, 3\}$, Divider needs at least $d^{\star} - \psi(x_i) + 1 + d^{\star} + d_i + 1 \ge 2d^{\star} + 2$ steps to move $D_{i}$ from $u_{i}^{\psi(x_i)}$ to $c_{j}^{r}$ via the shortest path $P[u_{i}^{\psi(x_i)}, u_{i}^{d^{\star} + 1}, d^{\star} - \psi(x_i)] \circ P[u_{i}^{d^{\star} + 1}, c_{j}^{r}, d^{\star} + d_i]$.
Note that, by the construction, the Divider's agents that are at distance less than or equal to $2d^{\star} + 2$ from $c_{j}^{r}$ are $D_{1}$, $D_{2}$ and $D_{3}$, only.
Facilitator can move Romeo and Juliet to $c_{j}^{r}$ in $2d^{\star} + 2$ steps through the paths $P[s, c_{j}^{r}, 2d^{\star} + 1]$ and $P[t, c_{j}^{r}, 2d^{\star} + 1]$, respectively.
Since Facilitator takes the first turn, she can move Romeo and Juliet to $c_{j}^{r}$ before Divider's agents.
Hence, Facilitator wins in $2d^{\star} + 2$ steps.

Suppose $\psi(x_1) > d_1, \psi(x_2) > d_2, \psi(x_3) > d_3$  (i.e. Second statement is \true).
Then, as mentioned in the Facilitator's strategy, her choice will be to move Romeo and Juliet towards $c_{j}^{\ell}$.
For $i \in \{1, 2, 3\}$, Divider needs at least $\psi(x_i) - 1 + 1 + 2d^{\star} - d_i + 1 > 2d^{\star} + 1$ steps to move $D_i$ from $u_i^{\psi(x_i)}$ to $c_{j}^{\ell}$ via the shortest path $P[u_{i}^{\psi(x_i)}, u_{i}^{0}, \psi(x_i) - 1] \circ P[u_{i}^{0}, c_{j}^{\ell}, 2d^{\star} - d_i]$.
Once again, by the construction, the Divider's agents that are at distance less than or equal to $2d^{\star} + 2$ from $c_{j}^{\ell}$ are $D_{1}$, $D_{2}$ and $D_{3}$.
Facilitator moves Romeo and Juliet to $c_{j}^{\ell}$ in $2d^{\star} + 2$ steps through the paths $P[s, c_{j}^{\ell}, 2d^{\star} + 1]$ and $P[t, c_{j}^{\ell}, 2d^{\star} + 1]$ respectively.
Since Facilitator takes the first turn, Romeo and Juliet is moved to $c_{j}^{\ell}$ before Divider agents and Facilitator wins in $2d^{\star} + 2$ steps.

This implies that if $(\calX, \calD, \calC)$ is a \no-instance of \textsc{(Monotone) NAE-Integer-$3$-Sat}, then Facilitator wins in at most $2d^{\star} + 2$ steps against Divider with $n + 2$ agents, i.e., $(G, s, t, n + 2)$ is a \yes-instance of \textsc{Rendezvous}.
\end{proof}

By the construction, the number of agents is upper bounded by the number of variables in \textsc{(Monotone) NAE-Integer-$3$-Sat} plus two.
Consider the set $S := \bigcup_{i \in [n]}\{u_{i}^{0},u_{i}^{d^{\star} + 1}\} \cup \{s,t\}$ of $2n + 2$ vertices in $G$.
It is easy to verify that $G - S$ is a collection of paths (corresponding to variable gadgets) and subdivided stars (centered at the vertices added while encoding the clauses).
It is easy to verify that the pathwidth of a subdivided star is at most two.
Hence, the feedback vertex set number and the pathwidth of the resulting graph are bounded by the linear function in the number of variables.
Lemma~\ref{lemma:fvs-forward-correct}, Lemma~\ref{lemma:fvs-backward-correct} and the fact that the reduction can be completed in the polynomial time in the size of input imply Theorem~\ref{thm:fvs-w-hard} which we restate here.
\fvswhard*

%% file: vertex-cover.tex
\section{Parameterizing by Vertex Cover}
\label{sec:vertex-cover}

In this section we focus on~\Cref{thm:vc-fpt-no-poly}:

\vcfptnpk*

Throughout this section, we assume that a vertex cover $X$ of size $\vc(G)$ is given as a part of the input. We first discuss the \FPT{} result.

\begin{reduction rule}
\label{rr:trivial-yes}
Consider an instance $(G,X,s,t, k)$ of \textsc{Rendezvous}.
If $s = t$, $st \in E(G)$, $|N(s) \cap N(t)| > k$, then return a trivial \yes-instance.
\end{reduction rule}

For the rest of this discussion, we will assume that any instance $(G, s, t, k)$ of \textsc{Rendezous} under consideration does \emph{not} satisfy the premise of \Cref{rr:trivial-yes}, i.e, we assume that we are not dealing with trivial \yes{} instances.
Also, since the vertices $s$ and $t$ can always be added to the vertex cover and this only increases the parameter by two, we assume for simplicity --- and without loss of generality --- that $s,t \in X$.

We now introduce some notation.
For a subset $Y \subseteq X$, let $I_Y \subseteq G \setminus X$ denote the set of vertices in $G \setminus X$ whose neighborhood is exactly $Y$.
Note that $\{I_Y\}_{Y \subseteq X}$ is a partition of $G \setminus X$ into at most $2^{\vc(G)}$ many parts.
For a vertex $v \in G \setminus X$, we use $\mathcal{E}_{G,X}(v)$ to denote the part that $v$ belongs to, in other words, $\mathcal{E}_{G,X}(v) = I_{N(v)}$.
We now apply the following reduction rule.

\begin{reduction rule}
\label{rr:eqclass}
Consider an instance $(G,X,s,t,k)$ of \textsc{Rendezvous}.
Repeat the following for each $v \in G \setminus X$.
If $|\mathcal{E}_{G,X}(v)| > k+1$, then choose any subset of exactly $k+1$ vertices from $\mathcal{E}_{G,X}(v)$ and delete rest of the vertices from $\mathcal{E}_{G,X}(v)$.
\end{reduction rule}

\begin{lemma}
    \label{lemma:rr-eq-safe}
    \Cref{rr:eqclass} is safe.
\end{lemma}

\begin{proof}
Let $(G,X,s,t,k)$ denote the input instance, and let $v \in G \setminus X$ be arbitrary but fixed.
Further, let $(H,X,s,t,k)$ denote the instance obtained by applying~\Cref{rr:eqclass} with respect to $v$.
If $|\mathcal{E}(v)| \leq k+1$ in $G$ then $G = H$ and there is nothing to prove.
Otherwise, let $Q_v \subseteq \mathcal{E}_{G,X}(v)$ denote the set of vertices deleted by the application of the reduction rule with respect to $v$.
Note that $H = G \setminus Q_v$.
Also observe that $|\mathcal{E}_{H,X}(v)| = k+1$.

To begin with, suppose the Facilitator has a winning strategy in $G$.
Observe that the Facilitator can employ the same strategy in $H$ as well, except when the strategy involves moving to a vertex $u \in Q_v$.
However, since $|\mathcal{E}_{H,X}(v)| = k+1$, we have that there is at least one vertex $w$ in $H \setminus X$ that has the same neighborhood as $u$ and is not occupied by an agent of the Divider, since the Divider has only $k$ agents at their disposal.
The strategy, at this point, would remain valid if we were to replace $u$ with $w$.
If the strategy involved using two distinct vertices from $Q_v$ in the same step, then note that we can modify the strategy and have the Faciliator's agents meet immediately at the vertex $w$.

On the other hand, if the Facilitator had a winning strategy in $H$, then it is easy to check that the Facilitator can win in $G$ by mimicing the strategy directly.
Another way to see this is the following.
Suppose that the Divider had a winning strategy in $G$.
Then observe that in any step, without loss of generality, if the Divider's agents occupy some vertices of $\mathcal{E}_{G,X}(v)$, we can replace this configuration with all of these agents on a single vertex of $\mathcal{E}_{G,X}(v)$ outside $Q_v$.
Thus any winning strategy for the divider in $G$ can be adapted to a valid winning strategy in $H$.
This concludes the argument for the equivlance of the two instances.
\end{proof}

\begin{lemma}
\label{lemma:vc-fpt}
{\sc Rendezvous} is \FPT\ when parameterized by the vertex cover number and the solution size.
\end{lemma}
\begin{proof}
    Observe that repeated applications of~\Cref{rr:eqclass} ensures that $|V(G)| = |X| + |G \setminus X| \leq \vc(G) + 2^{\vc(G)} \cdot (k+1)$. Thus we have an exponential kernel in $\vc(G)$, and the claim follows.
\end{proof}

%% file: vc-npk.tex
Now, we establish the lower bound claimed in~\Cref{thm:vc-fpt-no-poly} by showing the following.

\begin{lemma}
\label{lemma:vc-no-poly}
{\sc Rendezvous} does not admit a polynomial kernel when parameterized by the vertex cover number and the solution size unless \NP $\subseteq$ \co-\NP/poly.
\end{lemma}

The proof is based on observing that the instance in {the} reduction {used} in~\cite{DBLP:conf/wg/FominGT21} --- {to prove that problem is \co-\W[2]-\hard when parameterized by the solution size} --- has bounded vertex cover number. In particular, the reduction is from \textsc{Set Cover}, which does not admit a poylnomial kernel parameterized by the solution size and the size of the universe unless \NP $\subseteq$ \co-\NP/poly~\cite{kernelbook}. We reproduce the construction here for completeness.

\begin{proof}
\begin{figure}
\includegraphics[scale=0.25]{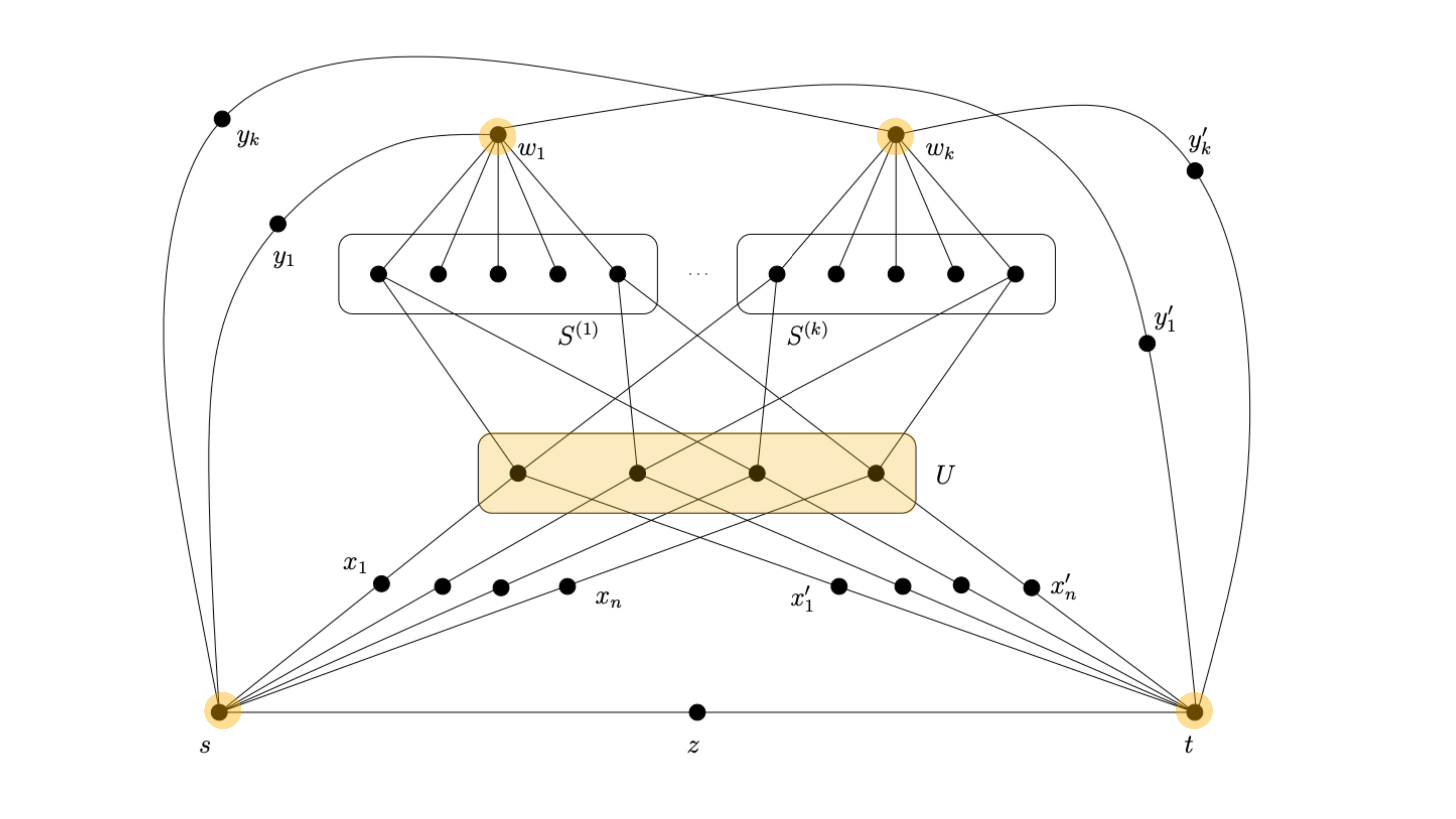}
\caption{A schematic of the reduction with the vertex cover vertices highlighted.}
\end{figure}

Recall that an instance of \textsc{Set Cover} consists of a universe $U$ of size $n$, a family over $U$ of size $m$, and a budget of $k$; and the question is if there exists a collection of at most $k$ sets from the given family whose union is $U$. Let $(U, \mathcal{S}, k)$ be an instance of \textsc{Set Cover}. Let $U=\left\{u_{1}, \ldots, u_{n}\right\}$ and $\mathcal{S}=\left\{S_{1}, \ldots, S_{m}\right\}$.

\begin{itemize}
\item Construct a set of $n$ vertices $U=\left\{u_{1}, \ldots, u_{n}\right\}$ corresponding to the universe.
\item For every $i \in\{1, \ldots, k\}$, construct a set of $m$ vertices $S^{(i)}=\left\{s_{1}^{(i)}, \ldots, s_{m}^{(i)}\right\} ;$ each $S^{(i)}$ corresponds to a copy of $\mathcal{S}$.
\item For every $i \in\{1, \ldots, k\}, h \in\{1, \ldots, m\}$ and $h \in\{1, \ldots, n\}$, make $s_{j}^{(i)}$ and $u_{h}$ adjacent if the element of the universe $u_{h}$ is in $S_{j} \in \mathcal{S}$.
\item For every $i \in\{1, \ldots, k\}$, construct a vertex $w_{i}$ and make it adjacent to $s_{1}^{(i)}, \ldots, s_{m}^{(i)}$.
\item Construct two vertices $s$ and $t$.
\item For every $h \in\{1, \ldots, n\}$, join $s$ and $u_{h}$ by a path $s x_{h} u_{h}$ and joint $u_{h}$ and $t$ by a path $u_{h} x_{h}^{\prime} t$.
\item For every $i \in\{1, \ldots, k\}$, join $s$ and $w_{i}$ by a path $s y_{i} w_{i}$ and join $w_{i}$ and $t$ by a path $w_{i} y_{i}^{\prime} t$.
\item Construct a vertex $z$ and make it adjacent to $s$ and $t$.
\end{itemize}

It is shown in~\cite{DBLP:conf/wg/FominGT21} that $(U, \mathcal{S}, k)$ is a yes-instance of SET Cover if and only if Divider with $k+1$ agents can win in the Rendezvous game. It is straightforward to check that $U \cup \{s,t\} \cup \{w_i ~|~ i \in [k]\}$ is a vertex cover for the reduced instance of size at most $m + k + 2$. The claim follows from the hardness of obtaining a polynomial kernel for \textsc{Set Cover} parameterized by $m+k$, since the equivalence of the instances is already known.
\end{proof}

%% file: polynomial-cases.tex
\section{Some Polynomial Cases}
\label{sec:polycases}

In this section, we focus on a few tractable scenarios.

\polytimecases*

We first discuss grids.

\begin{proposition}
    \label{prop:grid-graph-dynamic-sep}
    For a grid $G$ and two non-adjacent vertices $s, t \in V(G)$, $d_G(s, t) = 2$.
\end{proposition}

    \begin{proof}
    Consider an instance $(G,s,t,k)$ of \textsc{Rendezvous} where $G$ is a $M \times N$ undirected grid.
    Without loss of generality, we assume that $s$ and $t$ are non-adjacent in $G$.
    It is known~\cite[Theorem 2]{DBLP:conf/wg/FominGT21} that $d_{G}(s, t) = 1$ if and only if $\lambda_{G}(s, t) = 1$.
    Since in grid graph any two vertices are part of at least one cycle, $\lambda_{G}(s, t) \ge 2$.
    Hence, for any non-adjacent pair of vertices $s$ and $t$, $d_{G}(s, t) \ge 2$.
    Therefore, it is sufficient to show that $d_{G}(s, t) \le 2$.
    We prove that Divider with $2$ agents has a winning strategy on $G$ against Facilitator starting from $s$ and $t$.

    We respresent vertex $v$ of $G$ that is in $i$th row and $j$th column as $(i, j)$, where $i \in [M]$ is the row number of vertex $v$ and $j \in [N]$ is the column number of vertex $v$.
    Let $s$ be $(s_{x}, s_{y})$ and $t$ be $(t_{x}, t_{y})$, where $s_{x}, t_{x} \in [M]$, and $s_{y}, t_{y} \in [N]$.
    For Facilitator to win, she must make the difference between the row number as well as column number of the vertices having Romeo and Juliet equal to $0$.
    Since $(s_{x}, s_{y})$ and $(t_{x}, t_{y})$ are two different and non-adjacent vertices either $|s_{x} - t_{x}| > 0$ or $|s_{y} - t_{y}| > 0$.
    We assume without loss of generality $|s_{x} - t_{x}| > 0$ and $s_{x} < t_{x}$; in other words, $s$ and $t$ are on different rows and $s$ is ``below'' $t$ in the grid.

    We describe a winning strategy for Divider with the agents $D_{1}$ and $D_{2}$.
    Intuitively, the agent $D_1$ starts off to the ``top'' of $s$ and the agent $D_2$ starts off at a location to the ``bottom'' of $t$.
    Their goal will be to maintain the initial separation between $s$ and $t$ by not allowing the agent on $s$ to advance upwards or the agent on $t$ to advance downwards.
    They do this by ``tracking'' the agent movements and mimicing them whenever there is a shift to an adjacent column, and staying put if the agents are moving along the same column, in which case they are drifting further apart.

    In particular, to begin with, Divider puts $D_{1}$ in the vertex $(s_{x} + 1, s_{y})$ and $D_{2}$ in the vertex $(t_{x} - 1, t_{y})$ (since $s_{x} < t_{x}$, $s_{x} < M$ and $t_{x} > 1$).
    Then the following strategy is used.
    The agents $D_{1}$ and $D_{2}$ are keeping their positions until Facilitator moves Romeo or Juliet from $(s_{x}, s_{y})$ or $(t_{x}, t_{y})$, respectively.
    Whenever Facilitator moves Romeo, the agent $D_{1}$ replicates her move and similarly, whenever Facilitator moves Juliet, the agent $D_{2}$ replicates her move.
    Facilitator can move Romeo to either $(s_{x} - 1, s_{y})$ (if $s_{x} > 1$) or $(s_{x}, s_{y} - 1)$ (if $s_{y} > 1$ and $D_{2}$ is not on this vertex) or $(s_{x}, s_{y} + 1)$ (if $s_{y} < N$ and $D_{2}$ is not on this vertex).
    The vertex $(s_{x} + 1, s_{y})$ is occupied by $D_{1}$.
    Let the new position of Romeo be $(s^{\prime}_{x}, s^{\prime}_{y})$, where $s^{\prime}_{x} \in [M]$ and $s^{\prime}_{y} \in [N]$.
    Divider moves the agent $D_{1}$ to $(s_{x}, s_{y})$ or $(s_{x} + 1, s_{y} - 1)$ or $(s_{x} + 1, s_{y} + 1)$ corrosponding to the above mentioned three posibile moves of the Facilitator for Romeo.
    Similarly, Facilitator can move Juliet to either $(t_{x} + 1, t_{y})$ (if $t_{x} < M$) or $(t_{x}, t_{y} - 1)$ (if $t_{y} > 1$ and $D_{1}$ is not on this vertex) or $(t_{x}, t_{y} + 1)$ (if $t_{y} < N$ and $D_{1}$ is not on this vertex).
    The vertex $(t_{x} - 1, t_{y})$ is occupied by $D_{2}$.
    Let the new position of Juliet be $(t^{\prime}_{x}, t^{\prime}_{y})$, where $t^{\prime}_{x} \in [M]$ and $t^{\prime}_{y} \in [N]$.
    Divider moves the agent $D_{2}$ to $(t_{x}, t_{y})$ or $(t_{x} - 1, t_{y} - 1)$ or $(t_{x} - 1, t_{y} + 1)$ corrosponding to the above mentioned three posibile moves of the Facilitator for Juliet.
    Observe that, the difference of the row number of Juliet and Romeo does not decrease after any of the possible moves, i.e. $t^{\prime}_{x} - s^{\prime}_{x} \ge t_{x} - s_{x}$.
    Divider follows the same strategy after every move of Facilitator for Romeo and Juliet, and the strategy ensures that the difference of the row number of Juliet and Romeo does not decrease after any possible move of the Facilitator.
    Hence, Divider prevents Romeo and Juliet from meeting by ensuring that the difference of their row number does not decrease after any number of moves.
    This implies Divider wins.
    The argument also follows when $|s_{y} - t_{y}| > 0$ since Divider can prevent Romeo and Juliet from meeting by ensuring that the difference of their column number does not decrease after any number of moves.

    We conclude that Divider with $2$ agents has a winning strategy on $G$ against Facilitator starting from $s$ and $t$, which implies $d_{G} \le 2$.
    Since $d_{G} \le 2$ as well as $d_{G} \ge 2$, $d_{G} = 2$.
    This implies that for a grid graph $G$ and two non-adjacent vertices $s, t \in V(G)$, $d_G(s, t) = 2$.
    \end{proof}

We now turn to graphs of treewidth at most two. In this case, we show that $d_G(s,t) = \lambda_G(s,t)$, which leads to~\textsc{Rendezvous} being polynomially solvable on this class of graphs based on standard algorithms for computing $\lambda_G(s,t)$.

\begin{proposition}
If $G$ is a connected graph of tree-width at most $2$, then for every $s, t \in V(G)$, $d_G(s,t) = \lambda_G(s,t)$.
\end{proposition}
\begin{proof}
Consider an instance $(G,s,t,k)$ of \textsc{Rendezvous} where $G$ is a graph of treewidth at most $2$.
We recall that these are the series-parallel graphs, which are graphs with two distinguished vertices called terminals, formed recursively by two simple composition operations.
Specifically, we have the following definitions.
A two-terminal graph (TTG) is a graph with two distinguished vertices, $s$ and $t$ called source and sink, respectively.
The parallel composition $P_c = P_c(X,Y)$ of two TTGs $X$ and $Y$ is a TTG created from the disjoint union of graphs $X$ and $Y$ by merging the sources of $X$ and $Y$ to create the source of $P_c$ and merging the sinks of $X$ and $Y$ to create the sink of $P_c$.
The series composition $S_c = S_c(X,Y)$ of two TTGs $X$ and $Y$ is a TTG created from the disjoint union of graphs X and Y by merging the sink of $X$ with the source of $Y$. The source of $X$ becomes the source of $S_c$ and the sink of $Y$ becomes the sink of $S_c$.
A two-terminal series–parallel graph (TTSPG) is a graph that may be constructed by a sequence of series and parallel compositions starting from a set of copies of a single-edge graph $K_2$ with assigned terminals.
Finally, a graph is called series–parallel (SP-graph), if it is a TTSPG when some two of its vertices are regarded as source and sink.

The proof will proceed by induction on the number of vertices.
We will do a case analysis for sequence of compositions used to arrive at the final graph $G$.
In the base case, there is nothing to prove since $G$ is simply an edge.
We use $x$ and $y$ to denote the source and sink terminals, respectively.
For the induction hypothesis, we assume that the claim is true for all series-parallel graphs with less than $k$ vertices, where $k \geq 2$.
Now, let $G$ be a series-parallel graph having $k$ vertices.





Suppose $G$ is obtained by a series or parallel composition of graphs $G_1$ and $G_2$ and let $x$ and $y$ denote the source and sink terminals of $G$, while $x_b$ and $y_b$ denote the source and sink terminals of $G_b$ for $b \in \{1,2\}$. Note that $G_1$ and $G_2$ are series-parallel graphs having less than $k$ vertices.

{\color{SteelBlue}Case 1: $s \in G_1$ and $t \in G_2$; $s \neq x_1, s \neq y_1; t \neq x_2, t \neq y_2$}
\begin{itemize}
\item \textbf{Case 1A: The composition is series.} In this case static and dynamic separation number is one: $\{x\}$ or $\{y\}$ (the joined terminal).
\item \textbf{Case 1B: The composition is parallel.} Consider the path $s \rightarrow x \rightarrow t \rightarrow y \rightarrow s$. Since $s \neq t, s \neq x, s \neq y, t \neq x, t \neq y$, the considered path is a closed walk containing $s$ and $t$ which forms a cycle. So, $s$ and $t$ lies on a cycle. So, the lower bound on the dynamic separator is $2$. And the upper bound on the dynamic separator is also $2$ as the static separator in this case is $2$. So in this case static and dynamic separation number is two and is given by both terminals together: $\{x,y\}$.
\end{itemize}

{\color{SteelBlue}Case 2: $s \in G_1$ and $t \in G_1$ $s \neq x_1, s \neq y_1; t \neq x_1, t \neq y_1$}
\begin{itemize}
    \item \textbf{Case 2A: The composition is series.}

In this case, suppose $y_1$ and $x_2$ are identified as $g_{1, 2}$; and $x = x_1$ and $y = y_2$.


\begin{claim}
Static $s,t$ separators in $G_1$ will also work in $G$ and vice versa.
\end{claim}
\begin{proof}
\emph{Forward Direction.} Suppose $G_1$ has a static $(s, t)$ separator $S_1$ of size $k_1$. So, there does not exist any path from $s$ to $t$ in $G_1$ which does not contain any vertex of $S_1$.
Now, for the supergraph $G$, all the paths between $s$ and $t$ that does not pass through $G_2$ are already blocked by the static separator $S_1$.
Further, the paths that pass through $G_2$ will pass through the terminal vertex twice.
So these paths will be $s \rightarrow g_{1, 2} \rightarrow$ some vertices of $G_2 \rightarrow g_{1, 2} \rightarrow t$.
Suppose these paths are not blocked by $S_1$, then there also exist a path $s \rightarrow g_{1, 2} \rightarrow t$ in $G_1$ that are not blocked by $S_1$, which contradicts the assumption that $S_1$ is a static $(s, t)$ separator in $G_1$.
So, these paths are also blocked by $S_1$, which implies that $S_1$ is also the static separator of $G$.

\emph{Backward Direction.} Suppose $G$ has a static $(s, t)$ separator $S_1$ of size $k_1$.
Taking the vertices from $G_2$ in the static seperator can only block paths of the type $s \rightarrow g_{1, 2} \rightarrow$ some vertices of $G_2 \rightarrow g_{1, 2} \rightarrow t$.
So, $S_1$ can not contains more than one vertex from the graph $G_2$, else those vertices can be replaced by $g_{1, 2}$ which will result in a smaller sized static $(s, t)$ separator of $G$.
Hence, $S_1$ can contain at max one vertex from $G_{2}$.
Observe that if $S_1$ does not contain any vertex of $G_2$, then the same $S_1$ is also a static $(s, t)$ seperator in $G_1$.
If $S_1$ contains exactly one vertex from $G_2$, then that vertex can be replaced by $g_{1, 2}$ to form a same sized static $(s, t)$ seperator in $G_1$.
\end{proof}
\begin{claim}
$d_G(s,t) = \lambda_G(s,t) = k_1$.
\end{claim}
\begin{proof}
Suppose the claim is not true. Then we need fewer than $k_1$ guards in $G$, say $k_1 - 1$ guards are enough to separate $s$ from $t$ in $G$. But then this will also be a valid strategy in $G_1$, contradicting the induction hypothesis from which we know that $d_{G_1}(s,t) = \lambda_{G_1}(s,t) = k_1$.
\end{proof}
\item \textbf{Case 2B: The composition is parallel.}

In this case, we have that $y_1$ and $y_2$ join into $y$ and $x_1$ and $x_2$ join into $x$.

\begin{claim}
 $\lambda_{G_1}(s,t) \leq \lambda_{G}(s,t) \leq \lambda_{G_1}(s,t) + 1$.
\end{claim}
\begin{proof}

The first inequality follows from the fact that $G$ is a supergraph of $G_1$. For the second inequality, note that a separation of $s$ and $t$ in $G$ can be achieved by adding one of the terminal vertices to the static separator in $G_1$.
\end{proof}
\begin{claim}
$k_1 \leq d_G(s,t) = \lambda_G(s,t) \leq  k_1 + 1$.
\end{claim}
\begin{proof}
Suppose the claim is not true.
Suppose $\lambda_G(s, t) = k_1$, and we need fewer than $k_1$ guards in $G$, say $k_1 - 1$ guards are enough to separate $s$ from $t$ in $G$.
But then this will also be a valid strategy in $G_1$, contradicting the induction hypothesis from which we know that $d_{G_1}(s, t) = \lambda_{G_1}(s, t) = k_1$.
Suppose $\lambda_G(s, t) = k_1 + 1$, and we need fewer than $k_1 + 1$ guards in $G$, say $k_1$ guards are enough to separate $s$ from $t$ in $G$.
If $\lambda_G(s, t)$ has more than one vertex of $G_2$, then it will contradict the induction hypothesis from which we know that $\lambda_G(s, t) = k_1$.
If it does not have any vertex of $G_2$ then we can replace this separater with the union of static $(s, t)$ seperator of $G_1$ and the vertex $x$.
And if it contains exactly ine vertex of $G_2$, then that vertex can be replaced by one of the terminals $x$ or $y$ and it will still be a valid static $(s, t)$ seperator in $G$.
Observe that, to seperate $s$ and $t$ in $G$, at least one guard must be on one of the paths of type $s \rightarrow x \rightarrow$ some vertices of $G_2 \rightarrow y \rightarrow t$.
If not then there is no guard which can block $x$, $y$, and the vertices of $G_2$ and so it will not be a valid seperator.
Additionally, this vertex in not needed in the dynamic $(s, t)$ seperator of $G_1$.
Hence we can obtain a $k_1 - 1$ size dynamic $(s, t)$ separator in $G_1$ by eliminating this vertex, contradicting the induction hypothesis.
\end{proof}
\end{itemize}

{\color{SteelBlue}Case 3: $s$ is one of the terminals and $t \in G_1$ or $t \in G_2$; $s = x$ or $s = y; t \neq x_1, t \neq y_1$}
\begin{itemize}
    \item \textbf{Case 3A: The composition is series.}
    \begin{enumerate}
        \item If $s = x$ and $t \in G_2$ or if $s = y$ and $t \in G_1$, then the vertex at the junction of the composition is a separator of size one.
        \item If $s = x$ and $t \in G_1$ or $s = y$ and $t \in G_2$, then the static separator of $s$ and $t$ in $G$ is the static separator of $s$ and $t$ in $G_1$ or the static separator of $s$ and $t$ in $G_2$, of size $k_1$ or $k_2$ respectively.
    \end{enumerate}

\begin{claim}
If $s = x$ and $t \in G_1$, then $d_G(s,t) = \lambda_G(s,t) = k_1$, while if $s = y$ and $t \in G_2$, then $d_G(s,t) = \lambda_G(s,t) = k_2$.
\end{claim}
\begin{proof}
Consider the first statement and suppose the claim is not true.  Then we need fewer than $k_1$ guards in $G$, say $k_1 - 1$ guards are enough to separate $s$ from $t$ in $G$. But then this will also be a valid strategy in $G_1$, contradicting the induction hypothesis. The same argument works for the second statement as well.
\end{proof}
    \item \textbf{Case 3B: The composition is parallel.}
    \begin{claim}
    $\lambda_{G}(s,t) \leq \lambda_{G_1}(s,t) + 1$, if $t \in G_1$ and $\lambda_{G}(s,t) \leq \lambda_{G_2}(s,t) + 1$, if $t \in G_2$
    \end{claim}
    This argument in this case is analogous to Case 2B.
\end{itemize}

{\color{SteelBlue}Case 4. $s$ and $t$ are both terminals.}
\begin{itemize}
    \item \textbf{Case 4A: The composition is series.}
    \begin{enumerate}
        \item If $s = x; t = y$, then the vertex at the junction of the composition is a static separator of size one.
        \item If $s = x; t = z$, where $z$ denotes the the vertex at the junction of the composition, then a static separator of $s$ and $t$ in $G_1$ is also a static separator of $s$ and $t$ in $G$.
        \item If $s = z; t = y$, where $z$ is the same as before, then a static separator of $s$ and $t$ in $G_2$ is also a static separator of $s$ and $t$ in $G$.
    \end{enumerate}
\begin{claim}
    If $s = x$ and $t = z$, then $d_G(s,t) = \lambda_G(s,t) = k_1$, while if $s = z$ and $t \in y$, then $d_G(s,t) = \lambda_G(s,t) = k_2$.
\end{claim}
\begin{proof}
    Consider the first statement and suppose the claim is not true.
    Then we need fewer than $k_1$ guards in $G$, say $k_1 - 1$ guards are enough to separate $s$ from $t$ in $G$.
    But then this will also be a valid strategy in $G_1$, contradicting the induction hypothesis.
    The same argument works for the second statement as well.
\end{proof}
    \item \textbf{Case 4B: The composition is parallel, i.e, $s = x; t = y$.}\\
    In this case, note that the size of the static separator of $s$ and $t$ in $G$ = $\lambda_{G_1}(s,t)+\lambda_{G_2}(s,t)$. Indeed, any smaller subset would have either fewer than $\lambda_{G_1}(s,t)$ vertices in $G[V(G_1)]$ or fewer than $\lambda_{G_2}(s,t)$ vertices in $G[V(G_2)]$, implying the existence of an unblocked path from $s$ to $t$ via $G_1$ or $G_2$
\begin{claim}
$d_G(s,t) = \lambda_G(s,t) = k_1 + k_2$.
\end{claim}
\begin{proof}
Suppose the claim is not true. Then we need fewer than $k_1 + k_2$ guards in $G$, say $k_1 + k_2 - 1$ guards are enough to separate $s$ from $t$ in $G$. But then this will also be a valid strategy in $G_1$ $($or $G_2)$ with $< k_1$ $($or $< k_2)$ guards, contradicting the induction hypothesis.
\end{proof}
\end{itemize}
So, for all the cases $\lambda_G(s,t) = d_G(s,t)$, and this concludes our argument.
\end{proof}

%% file: conclusion.tex
\section{Conclusion}

In this work, we studied the game of rendezvous with adversaries on a graph introduced by Fomin, Golovach, and Thilikos~\cite{DBLP:conf/wg/FominGT21}.
The game is a natural dynamic version of the problem of finding a vertex cut between two vertices $s$ and $t$.
Given that the problem is W[2]-\hard\ when parameterized by the natural parameter, i.e.  the solution size,  we continued 
studying structural parameters of the input graph initiated by Fomin et al.  \cite{DBLP:conf/wg/FominGT21}. 
We proved, to our surprise,  that the problem is {\co-\NP-\hard} even when restricted to
graphs whose feedback vertex set number is at most $14$, or 
pathwidth is at most $16$.
In particular,  we proved {\sc Rendezvous} is {\co-\para-\NP-\hard} parameterized by treewidth, thereby answering an open question by Fomin et al.  \cite{DBLP:conf/wg/FominGT21}.
It turns out that even augmenting the feedback vertex set number or the pathwidth with the solution size is not enough.
Specifically,  we proved that {\sc Rendezvous} is {\co-\W[1]-\hard} when parameterized by 
the feedback vertex set number and the solution size, or
the pathwidth and the solution size.
Towards the positive side,  we proved  that the problem admits a natural exponential kernel when parameterized 
by the vertex cover number and the solution size, 
however this kernel cannot be improved to a polynomial kernel under standard complexity-theoretic assumptions.
Finally, we presented polynomial time algorithms on two restricted cases and proved that 
{\sc Rendezvous} can be solved in polynomial time on the classes of treewidth at most two graphs and grids.

While we addressed the structural parameterized by arguably the most well studied parameters, 
it remains interesting to study the parameterized complexity by other structural parameters.
Amongst these,  we highlight the following question:
\emph{Is {\sc Rendezvous} {\em \W[1]-\hard} when parameterized by the vertex cover number (only)?}
We tend to believe it is indeed the case.
To the best of our knowledge,  the problems that are  \W[1]-\hard\ when parameterized by the vertex cover number, like \textsc{List Coloring}, 
\textsc{Weighted $(k, r)$-Center}, etc.,  have additional input arguments like lists or weights. 
We believe that the dynamic natural of the {\sc Rendezvous} problem might make it an exception to the above known trend.